\documentclass[10pt,journal]{IEEEtran}

\IEEEoverridecommandlockouts
\usepackage{stfloats}
\usepackage{amsfonts}
\usepackage{amssymb}
\usepackage{amsthm}
\usepackage{cite}
\usepackage[cmex10]{amsmath}
\usepackage{float}
\usepackage{color}
\usepackage{algpseudocode}
\usepackage{multirow}
\usepackage[normalem]{ulem}
\usepackage{tabularx}
\usepackage{fancybox,dashbox}
\usepackage{bbm}
\usepackage{subfig}
\usepackage{cases}
\usepackage{bm}
\usepackage[ruled,lined,commentsnumbered]{algorithm2e}

\newtheorem{theorem}{\textit{Theorem}}

\newtheorem{proposition}{\textit{Proposition}}
\newtheorem{corollary}{\textit{Corollary}}
\newtheorem{definition}{\textit{Definition}}
\newtheorem{remark}{\textit{Remark}}

\def\BibTeX{{\rm B\kern-.05em{\sc i\kern-.025em b}\kern-.08em
 T\kern-.1667em\lower.7ex\hbox{E}\kern-.125emX}}

\ifCLASSINFOpdf
\usepackage[pdftex]{graphicx}
\DeclareGraphicsExtensions{.pdf,.jpeg,.png}
\else
\usepackage[dvips]{graphicx}
\DeclareGraphicsExtensions{.eps}
\fi

\begin{document}

\title{Optimizing Age of Information in Wireless Uplink Networks with Partial Observations}
\author{
Jingwei Liu, Rui Zhang, Aoyu Gong, and He Chen
\thanks{The work of H. Chen and J. Liu are supported in part by the Innovation and Technolgy Fund (ITF) under Project ITS/204/20 and the CUHK direct grant for research under Project 4055126. The work of R. Zhang is supported in part by the Research Talent Hub PiH/380/21 under Project ITS/204/20. 
This article was presented in part at GLOBECOM 2020.

J. Liu, R. Zhang and H. Chen are with Department of Information Engineering, The
Chinese University of Hong Kong, Hong Kong SAR, China (e-mail:lj020@ie.cuhk.edu.hk; ruizhang@ie.cuhk.edu.hk; 
he.chen@ie.cuhk.edu.hk).

A. Gong is with the School of Computer and Communication Sciences, \'Ecole Polytechnique F\'ed\'erale de Lausanne, Lausanne 1015, Switzerland (e-mail: aoyu.gong@epfl.ch).
}
}


\maketitle

\begin{abstract}
This paper considers a wireless uplink network consisting of multiple end devices and an access point (AP). Each device monitors a physical process with randomly generated status updates and sends these update packets to the AP in the uplink. The AP aims to schedule the transmissions of these devices to optimize the network-wide information freshness, quantified by the age of information (AoI) metric. Due to the stochastic arrival of the status updates at end devices, the AP only has \textit{partial observations} of system times of the latest status update packets at end devices when making scheduling decisions. Such a decision-making problem can be naturally formulated as a partially observable Markov decision process (POMDP). We reformulate the POMDP into an equivalent belief Markov decision process (belief-MDP), by defining fully observable belief states of the POMDP as the states of the belief-MDP.
The belief-MDP in its original form is difficult to solve as the dimension of its states can go to infinity and its belief space is uncountable. Fortunately, by carefully leveraging the properties of the status update arrival processes (i.e., Bernoulli processes), we manage to simplify the belief-MDP substantially, where every feasible state is characterized by a two-dimensional vector. Based on the simplified belief-MDP, we devise a low-complexity scheduling policy, termed Partially Observing Max-Weight (POMW) policy, for the formulated AoI-oriented scheduling problem.
We derive upper bounds for the time-average AoI performance of the proposed POMW policy. We analyze the performance guarantee for the POMW policy by comparing its performance with a universal lower bound available in the literature. Numerical results validate our analyses and demonstrate that the performance gap between the POMW policy and its fully observable counterpart is proportional to the inverse of the lowest arrival rate of all end devices.
\end{abstract}

\begin{IEEEkeywords}
Age of information, multiuser scheduling, partially observable Markov decision process, and belief Markov decision process.
\end{IEEEkeywords}


\section{Introduction}
\IEEEPARstart{T}{he} rapid development of wireless communication technologies in the past decades has stimulated their ubiquitous applications in time-critical systems, such as vehicular networks and industrial control networks \cite{6917404,7467436,7799033}. 
In these applications,
information (e.g., velocity and position of a vehicle) needs to be delivered to targeted receivers as timely as possible.
The stale information could cause severe consequences, e.g., damages to facilities or even losses of human lives.
Hence, the information timeliness or freshness in these networks is of great importance.
To quantify the information freshness, 
the age of information (AoI) metric has been proposed and extensively investigated in the literature (e.g., see \cite{kosta2017age,8930830,8954939,8406973,9013924,8406966,5984917,6195689,8065840,sun2019age,8006703,8972306,10.1145/3466772.3467040,7415972,7511331,dedhia2020you} and references therein).
More specifically, 
AoI is defined as the time elapsed since the generation of the last successfully received message at destination \cite{kosta2017age}.
Many efforts have been made on tackling transmission scheduling problems to minimize the time-average AoI of various network settings.
Early work focused on the AoI-based transmission scheduling problem in single-user networks, see e.g., \cite{8000687,8712546,8486307,8406846,8845182,8406909}, where the AoI performance of the single user was optimized by determining when to transmit a status update packet.
Recent work has shifted to design the AoI-based scheduling policies for multiuser networks, see e.g., \cite{8514816,8935400,8933047,9348022,8807257,8845083,9376717}.
In these work, the network-wide time-average AoI was optimized by determining how to schedule the transmission sequence of multiple users.

In downlink multiuser networks, an access point (AP) monitors multiple information sources and schedules transmissions of the generated status update packets from itself to the corresponding end devices, respectively.
In this context, the AP can completely know the evolution of AoI when acknowledgements are provided by end devices.
The AoI-based scheduling problems in downlink multiuser networks were thoroughly studied in \cite{9525063,8514816,8933047}.
The authors in \cite{8514816} considered the ``generate-at-will'' model for the generation of status updates.
In this model, the AP generates a status update for an information source whenever the transmission to its targeted end device is scheduled. 
As such,
the AP only needs to consider the instantaneous AoI values of all end devices when making scheduling decisions.
Authors in \cite{8514816} first proved that in symmetric networks, a greedy policy, which schedules the end device with the highest value of instantaneous AoI, is optimal for minimizing the long-term average AoI. For more general networks, three low-complexity scheduling policies were proposed and compared, including a Max-Weight policy derived from the Lyapunov optimization framework \cite{neely2010stochastic}, a randomized policy, and a Whittle’s Index policy.
Ref. \cite{8933047} extended the Max-Weight policy to the downlink networks with the ``stochastic arrival'' model, and an upper bound for the network-wide time-average AoI was derived. 
On the other hand, \cite{8935400} developed a Whittle's Index policy for the same scenario as in \cite{8933047}.
In the ``stochastic arrival'' model, the generation of status update packets for each information source follows a stochastic process. 
In this case,
the system times of update packets at the AP and the instantaneous AoI values of all end devices need to be jointly considered when designing the scheduling policies for the AP.




In uplink multiuser networks, on the other hand, 
each end device monitors the statuses of a separate information source and sends status update packets to a common AP.
The AP aims to maintain a low network-wide AoI performance by carefully scheduling the transmissions of status update packets in the uplink. As the information destination, the AP has a full track of the AoI values of all streams of status updates. 
For the “generate-at-will” model, each node will
generate a new status update packet once granted to transmit.
As such, the system times of status update packets are always equal to 1 and thus are perfectly known to the AP.

In this case, the AoI-oriented scheduling problems in uplink networks are mathematically equivalent to those in downlink networks when the scheduling constraints of the two types of networks are the same.
By contrast, when it comes to the ``stochastic arrival'' model, the scheduling problems in uplink multiuser networks are largely different from those in downlink networks.
This is because in uplink networks, the AP may need to make scheduling decisions under partial observations of the system times of randomly generated status update packets at end device side.
The complete observations of the system times of all status update packets requires end devices to report the arrivals of new status updates to the AP before each scheduling decision-making.
Such a reporting procedure could lead to considerable network overhead, especially when status update packets are short.
Therefore, it is of practical significance to devise scheduling policies for the AP that can be executed without the need of complete knowledge of the system times of status update packets at the end device side. 
In that case, the AP only has an observation of the system time of status update of a certain end device only when the device is scheduled to transmit and the transmission is successful.
To the best knowledge, such an AoI-based scheduling problem for uplink multiuser networks with partial observations has not been thoroughly studied in open literature.
We note that \cite{8935400} developed a Whittle’s Index policy for optimizing AoI in
an uplink multiuser network with the “stochastic arrival” model. However, the system times of status update packets at all nodes are assumed to be fully observed, making the scheduling problem mathematically equivalent to that in \cite{8933047}.

As an attempt to fill the gap, in this paper we aim to optimize the \textit{expected weighted sum AoI} for an uplink multiuser network with stochastic arrivals of status updates at end devices.
The arrivals of status update packets at end devices are assumed to follow independent Bernoulli processes, which is commonly used in the literature (see e.g., \cite{8514816,8933047,8935400,neely2010stochastic}). 
We consider that the end devices will not report the random arrivals of the status updates to the AP for minimizing the network overhead. As such, the designed scheduling policy needs to make decisions with partial observations. The main contributions of this paper are summarized as follows. 
\begin{itemize}
\item We formulate our AoI-oriented scheduling problem as a POMDP problem considering the incomplete knowledge of status update arrivals of end devices at the AP.
The instantaneous system times of status update packets at the end devices and the instantaneous AoI at the AP are jointly defined as the states
of the POMDP.
We reformulate the POMDP to an equivalent belief Markov decision process (belief-MDP), where
the states of the belief--MDP, termed belief states, are defined as the posterior distributions of the states of the POMDP.
We remark that computing the optimal policy for the belief-MDP (or the POMDP) is a PSPACE-complete problem \cite{papadimitriou1987complexity}, which is not practically computable.
Nevertheless, such a belief-MDP reformulation benefits the policy design and the theoretical analysis since the belief states characterize sufficient statistics of the system.
\item 
To solve the formulated belief-MDP, we propose an effective simplification to characterize all feasible infinite-dimensional belief states as two-dimensional vectors.
This is achieved by analyzing how Bernoulli arrival processes of status updates at end devices affect the evolution of the belief states.
By doing so, we reduce the continuous spaces of the belief states to discrete ones. That is, we extract the feasible belief spaces from the corresponding distribution spaces.
The simplification of belief updates in belief-MDP
largely facilitate the design of scheduling policies as well as the theoretical analysis of the scheduling policies' performance.
\item
We devise a low-complexity Partially Observable Max-Weight (POMW) policy, inspired by the Lyapunov optimization framework \cite{neely2010stochastic}. The POMW policy aims to minimize a Lyapunov Drift function, defined as the expectation of the sum of weighted instantaneous AoI, in each time slot under condition of the current belief states. 
Based on the simplified belief-MDP model and a Randomized Scheduling policy proposed in \cite{8933047}, we derive upper bounds for the expected weighted sum AoI performance of the POMW policy. 
Further, we evaluate the performance guarantee for the POMW policy, which is defined as the ratio between the AoI performance of the POMW policy and that of a universal lower bound.
Simulation results validate our theoretical analysis. Simulation results also show that the performance gap between the POMW policy and its fully observable counterpart is inversely proportional to the lowest arrival rate of all end devices.
Moreover, the proposed POMW policy is superior to the baseline policies, which do not use the statistical information of the system times of the status update packets at end devices.
\end{itemize}

We notice a handful of efforts on designing AoI-oriented scheduling policies that also considered networks with partial observations \cite{8712546,9517880,8404794,9590543}.
Leng and Yener investigated the AoI minimization in a time-slotted cognitive radio energy harvesting network \cite{8712546}. 
In \cite{8712546}, a secondary user decides whether to send a status update in each time slot with the partially observable occupation status of the spectrum. In this context, the AoI minimization problem was formulated as a POMDP.
The optimal policy with threshold structure was sought by dynamic programming (DP).
In \cite{9517880}, the authors formulated the AoI optimization problem of a status update system with a partially observable Gilbert–Elliott Channel as a belief-MDP.
The authors developed an efficient structure-aware algorithm that is shown to be near-optimal.
Sert and Elif et al. \cite{8404794} investigated an AoI-based minimization on real-life TCP/IP connections with unknown delay and service time distributions.
They trained a Deep Q-network (DQN) algorithm to perform actions on the network and obtained a near-optimal AoI performance.
Shao and Liew et al. \cite{9590543} focused on the minimum-age scheduling for a time-slotted wireless uplink network, where multiple sensors are used to monitor one common physical process.
The authors formulated a POMDP and analyzed the performance of a greedy policy where an AP schedules the sensor with the minimum system time in each slot.
All of the above work considered the AoI-based scheduling problem with one stream of status update. As such, the developed methods cannot be directly applied to solve our scheduling problem with multiple streams of status updates, where we need to deal with the intricate interactions of the AoI evolutions of multiple end devices.


We remark that part of the results presented in this work has been published in the conference version \cite{9348022}. In \cite{9348022}, we formulated the considered scheduling design problem as a POMDP and solved it by directly applying the classical DP method. A low-complexity myopic policy was also proposed. However, the complication of the problem in its default form stopped us from conducting any theoretical analysis. In the current work, we reformulate the POMDP into a belief-MDP and put forth an effective simplification of the belief-MDP. Such simplification substantially facilitate the design of the POMW policy as well as the theoretical analysis of its performance.


\textbf{\textit{Notations:}} In this paper, $\mathbb{Z}^+$ denotes the set of non-negative integers, $\mathbb{E}[\cdot]$ denotes the operator of expectation, $[\cdot]$ denotes the representation of a vector containing the same type of elements, $\left\langle\cdot\right\rangle$ denotes a tuple containing different types of elements, and $\rVert\cdot\rVert_1$ denotes the $l_1$-norm of a vector.
For two vectors, $\bm{v}=[v_l]_{l=1}^L$ and $\bm{w}=[w_l]_{l=1}^L$, with the same dimension $L$, $\bm{v}\ge\bm{w}$ represents $v_l\ge w_l,\forall l$.

\section{System Model and POMDP Formulation}\label{SectionSys}


\subsection{System Model}

As shown in Fig. \ref{Network}, we consider a multiuser wireless uplink network consisting of one access point (AP) and $N$ status-updating end devices.
Those end devices are also called nodes hereafter, and indexed by $i \in \{ 1,\dots,N \}$. The considered system is time-slotted, and the time slot is indexed by $t \in \mathbb{Z}^+$. We consider a stochastic arrival model for the status update packets at each node. Specifically, the status update arrival at node $i$ in each slot follows an independent and identically distributed (i.i.d.) Bernoulli process\footnotemark{} with an arrival rate $\lambda_i$. 
\footnotetext{The extension to the case with Markovian packet arrival processes will be discussed in Remark \ref{RemarkM}.}
Each node maintains a single buffer to store the latest status update. That is, the current status update in the buffer will be replaced once a new one arrives. Such a single-buffer configuration, equivalent to the last-come-first-served (LCFS) queuing model, has been shown to achieve the best information freshness performance in stochastic arrival models \cite{6195689,8933047}. 
All nodes share a common wireless channel, and their transmissions of the status update packets in the uplink are coordinated by the AP. Specifically, at the beginning of each slot, the AP grants one node to transmit its latest status update packet. 
We denote the scheduling indicator for node $i$ in slot $t$ by $a_{t,i}\in\left\{0,1\right\}$, which is equal to 1 when node $i$ is scheduled to transmit in slot $t$, and $a_{t,i}=0$ otherwise.
Only one node is scheduled to transmit in each slot, thus the transmission collision among nodes is avoided. The transmission of each status update packet takes one time slot. We further assume that the transmission from node $i$ to the AP is error-prone with a time-invariant successful rate $p_i$.

\begin{figure}[t]
	\centering
	\includegraphics[width=0.5\textwidth]{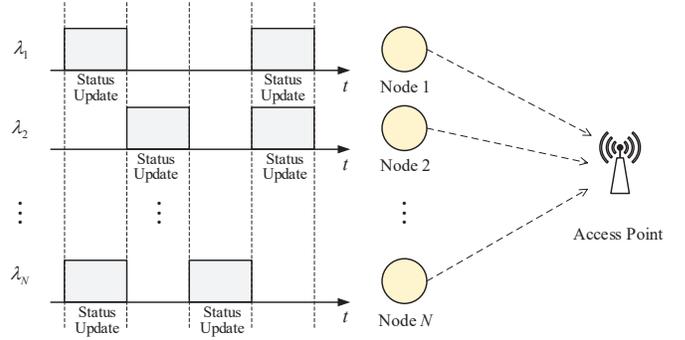}
	\caption{The multiuser uplink system with stochastic arrival of status updates.}
	\label{Network}
\end{figure}

\subsection{Information Freshness Metric}\label{SectionAoI}
We adopt the AoI metric, originally proposed in \cite{5984917}, to quantify the information freshness of all nodes at the AP. 
To characterize the AoI mathematically, we first define the local age $d_{t,i}$, which measures the system time of the last arrived status update packet at node $i$ in slot $t$. If there is no arrival of status update at node $i$ in the current slot, the local age of the $i$-th node will increase by $1$ at the beginning of next slot. Otherwise, the packet stored at the node is replaced by the newly arrived one, and its local age is reset to $1$ at the beginning of next slot. Therefore, the evolution of $d_{t,i}$ is given~by
\begin{equation}
	d_{t+1,i}=\begin{cases}
		1, &\text{if status update arrives at node}\ i\\
		&\text{in slot}\ t,\\
		d_{t,i} +1, &\text{otherwise}.
	\end{cases}
\end{equation}

If node $i$ is scheduled to transmit at the beginning of slot $t$ and its transmission is successful, the local age of node $i$ will be observed by the AP. As such, the destination AoI of node $i$, denoted by $D_{t,i}$, will be set to $d_{t,i}+1$ at the beginning of the next slot.
Otherwise, if node $i$ is not scheduled or the transmission fails, $D_{t,i}$ will increase by $1$ at the beginning of the next slot. Mathematically, the evolution of $D_{t,i}$ is given~by

\begin{equation}\label{eqAoI}
	\begin{split}
	  &D_{t+1,i}\\
	  &=\begin{cases}
		d_{t,i} +1, 
		&\text{if the status update of node $i$ is}\\
		&\text{successfully received by the AP in slot}\ t,\\
		D_{t,i} +1, &\text{otherwise}.
	\end{cases}
	\end{split}
\end{equation}
In this paper, we assume that the local age and the destination AoI of each node are initialized as $1$
, i.e., $d_{0,i}=D_{0,i}=1,\forall i$.

We remark that the local age $d_{t,i}$ and the destination AoI $D_{t,i}$ evolve independently across nodes.
We consider that the AP does not grasp the specific evolutions of the local ages at all nodes and it only has the statistical arrival information (i.e., the values of $\lambda_i$'s). Otherwise, the nodes need to notify each of their status update arrivals to the AP, which will lead to considerable network overhead, especially when the status update packets are relatively short. In this context, the AP only has an observation of the local age of a particular node once the node is scheduled and the transmission succeeds. Nevertheless, the AP can track the destination AoI values of all nodes, no matter whether they are scheduled or not. Overall, the AP has full information of the AoI $D_{t,i}$'s and partial observations of the local age $d_{t,i}$'s when making scheduling decisions.

\subsection{POMDP Formulation}\label{POMDPsec}

In this work, we adopt the long-term expected weighted sum AoI (EWSAoI) as the performance metric, which is mathematically defined as
\begin{equation}
    \lim_{T\to \infty}\frac{1}{NT}\mathbb{E}\left[\sum^T_{t=1}\sum^N_{i=1}\omega_iD_{t,i}\Big|\pi\right],
\end{equation} 
where $\omega_i\in(0,\infty)$ denotes the weight coefficient of node $i$, the expectation is taken over all system dynamics, and $\pi$ denotes a given multiuser scheduling policy.
We aim to devise a scheduling policy $\pi$ for the AP to minimize the long-term EWSAoI while fulfilling the scheduling constraint. 
Mathematically, we have the following optimization problem
\begin{equation}
    \begin{split}
    \min_\pi\quad & \lim_{T\to \infty}\frac{1}{NT}\mathbb{E}\left[\sum^T_{t=1}\sum^N_{i=1}\omega_iD_{t,i}\Big|\pi\right],\\
    \mbox{s.t.,}\quad &
    \sum^N_{i=1}a_{t,i}\le 1,\forall t,
    \end{split}
\end{equation}
where the scheduling constraint is that the AP can schedule at most one node in each slot.
In our design,
the AP makes the scheduling decision
at the beginning of each time slot. 
\textit{The information available at the AP} for decision making includes the values of $\lambda_i$'s, $p_i$'s, $\omega_i$'s, the full observations of the destination AoI $D_{t,i}$'s, and the partial observations of the local age $d_{t,i}$'s.
Such a decision-making problem with partial observations is naturally formulated as a POMDP with the following components:

\begin{itemize}
\item \underline{\textit{States.}} The state of node $i$ in slot $t$ is denoted by $\bm{s}_{t,i}\triangleq \left\langle d_{t,i},D_{t,i}\right\rangle$, where $d_{t,i},D_{t,i} \in \mathbb{Z}^+$. Then, the network-wide state in slot $t$ is denoted by $\bm{s}_t \triangleq \left\langle \bm{d}_t,\bm{D}_t\right\rangle$, where $\bm{d}_t\triangleq \left[ d_{t,1},d_{t,2},\dots,d_{t,N}\right] \in \bm{\mathcal{D}}\triangleq \left( \mathbb{Z}^+\right)^N $ and $\bm{D}_t\triangleq \left[ D_{t,1},D_{t,2},\dots,D_{t,N}\right] \in \bm{\mathcal{D}}$, respectively. In addition, we denote the spaces of $\bm{s}_{t,i}$ and $\bm{s}_t$ by $ \bm{\mathcal{S}}_i\triangleq\left\lbrace \bm{s}_{t,i}|D_{t,i}\geq d_{t,i}\right\rbrace $ and $ \bm{\mathcal{S}}\triangleq\left\lbrace \bm{s}_t|\bm{D}_t\geq\bm{d}_t\right\rbrace $, respectively.
	
\item \underline{\textit{Actions.}} The network-wide action in slot $t$ is denoted by $\bm{a}_t\triangleq \left[ a_{t,1},a_{t,2},\dots,a_{t,N} \right]$.
Recall that AP schedules at most one node in each slot, hence we have $\left|\bm{a}_t\right|\le 1$. Denote by $\bm{\mathcal{A}}$ the space of all actions, we have $\bm{a}_t\in\bm{\mathcal{A}}$.
\item \underline{\textit{Observations.}} 
We denote the network-wide observation of the state of the nodes by $\bm{o}_t\triangleq \left[ \bm{o}_{t,1},\bm{o}_{t,2},\dots,\bm{o}_{t,N} \right] \in \bm{\mathcal{O}}$, where $\bm{\mathcal{O}}$ is the space of all observations. 
Specifically, $\bm{o}_{t,i} \triangleq \left\langle D_{t,i},\hat{d}_{t,i} \right\rangle$ is the observation of node $i$ in slot $t$, consisting of the full-observed destination AoI, $D_{t,i}$, and the partial-observed local age $\hat{d}_{t,i} $. 
We have $\hat{d}_{t,i}\in \mathbb{Z}^+ \bigcup \left\{X\right\}$, where $X$ denotes no observation of the local age of node $i$ when the node is not scheduled or the node is scheduled but the transmission fails. With these new notations, $\bm{o}_t$ can be denoted by $\left\langle \bm{D}_t, \hat{\bm{d}}_t\right\rangle$, where $\hat{\bm{d}}_t=\left[ \hat{d}_{t,1},\dots,\hat{d}_{t,N} \right] $.
	
\item \underline{\textit{Transition Function.}} We define the transition probability of network-wide states as $\Pr\left(\bm{s}_{t+1}|\bm{s}_t,\bm{a}_t\right)$, which denotes
the conditional probability of state $\bm{s}_{t+1}$ given state $\bm{s}_t$ and action $\bm{a}_t$. We note that the transitions of $\bm{D}_t$ and $\bm{d}_t$ are conditionally independent of each other and the transition of the local age $d_t$ is independent of the action $a_t$. We then have
\begin{equation}
	\Pr\left(\bm{s}_{t+1}|\bm{s}_t,\bm{a}_t \right)= \Pr\left( \bm{D}_{t+1}|\bm{s}_t,\bm{a}_t\right)\Pr\left( \bm{d}_{t+1}|\bm{d}_t\right), 
\end{equation}
where
\begin{equation}
	\label{eq_D}
	\Pr\left( \bm{D}_{t+1}|\bm{s}_t,\bm{a}_t \right) =\prod_{i=1}^{N}\Pr\left(D_{t+1,i}|\bm{s}_{t,i},a_{t,i} \right), 
\end{equation}
and
\begin{equation}
	\label{eq_E}
	\Pr\left( \bm{d}_{t+1}|\bm{d}_t \right) =\prod_{i=1}^{N}\Pr\left(d_{t+1,i}|d_{t,i}\right).
\end{equation}
We can further express each term on the right-hand side of \eqref{eq_D} as
\begin{equation}
	\begin{split}
		&\Pr\left( D_{t+1,i}|\bm{s}_{t,i},a_{t,i} \right)\\
			&=
	\begin{cases}
		p_i,& \text{if}\ a_{t,i}=1,\text{and}\ D_{t+1,i}=d_{t,i}+1, \\				
		1-p_i,& \text{if}\ a_{t,i}=1,\text{and}\ D_{t+1,i}=D_{t,i}+1,\\
		1,& \text{if}\ a_{t,i}=0,\text{and}\ D_{t+1,i}=D_{t,i}+1,\\
		0,& \text{otherwise}.
	\end{cases}
	\end{split}
\end{equation}
	
Similarly, for each term on the right-hand side of \eqref{eq_E}, we have
\begin{equation}\label{Trans_d}
	\Pr\left( d_{t+1,i}|d_{t,i}\right)=
	\begin{cases}
		\lambda_i, &\text{if}\ d_{t+1,i}=1,\\
		1-\lambda_i, &\text{if}\ d_{t+1,i}=d_{t,i}+1,\\
		0, & \text{otherwise.}
	\end{cases} 
\end{equation}

	
\item \underline{\textit{Observation Function.}}
Denote by $\Pr\left( \bm{o}_t|\bm{s}_t,\bm{a}_t\right) $ the network-wide observation function, which is defined as the probability of observation $\bm{o}_t$ conditioned on state $\bm{s}_t$ and action $\bm{a}_t$.
Note that $\bm{D}_t$ is fully observable at the AP and the evolution of $\hat{d}_{t,i}$ with different $i$ are independent from each other. We thus have
\begin{equation}
	\begin{split}
	  \Pr\left( \bm{o}_t|\bm{s}_t,\bm{a}_t\right)
	  &= \Pr\left( \hat{\bm{d}}_t|\bm{d}_t,\bm{a}_t\right)\\
	  &=\prod_{i=1}^{N}\Pr\left( \hat{d}_{t,i}|d_{t,i},a_{t,i}\right),
	\end{split}
\end{equation}
where we term
\begin{equation}
  \begin{split}
    &\Pr\left( \hat{d}_{t,i}|d_{t,i},a_{t,i}\right)\\
	  &=
	  \begin{cases}
		  p_i , &\text{if}\ \hat{d}_{t,i}=d_t\ \text{and}\ a_{t,i}=1,\\
		  1-p_i, &\text{if}\ \hat{d}_{t,i}= X \ \text{and}\ a_{t,i}=1,\\	
		  1, &\text{if}\ \hat{d}_{t,i}= X \ \text{and}\ a_{t,i}=0,\\	
		  0, &\text{otherwise.}
	  \end{cases}
  \end{split}
\end{equation}
as the local age observation function of node $i$.

\item \underline{\textit{Immediate Reward.}}
We target to
optimize the long-term EWSAoI. 
Based on that, We define the immediate reward of state $\bm{s}_t$ as $r\left( \bm{s}_t\right)\triangleq \sum^N_{i=1} \omega_i D_{t,i} $.
	
\end{itemize}


We remark that due to the partially observed network-wide state $\bm{s}_t$, the formulated POMDP problem cannot be solved by directly applying the existing AoI-oriented scheduling frameworks designed for the scenarios with full observation of network-wide states (e.g., \cite{8486307,8933047,8514816}). 
To circumvent the problem, we will leverage the sufficient posterior probability distribution of $\bm{s}_t$ with the observation $\bm{o}_t$ at the AP. 
Such probability distributions are also named as the belief states of the POMDP \cite{papadimitriou1987complexity}.
In the following, we will reformulate our POMDP as a belief-MDP, where the belief states of the POMDP are regarded as the states of the belief-MDP.

\section{Belief-MDP Formulation and Simplification}\label{POMDP}
 
In this section, we first reformulate the POMDP introduced in Section \ref{SectionSys} as a belief-MDP and then simplify the belief-MDP to gain more insights.

\subsection{Reformulation of the POMDP}
With reference to \cite{mcallester2013approximate}, a POMDP can be converted to an equivalent belief-MDP based on the belief states of the system.
To that end, we now introduce the definitions of the belief states and other components of the belief-MDP version of our POMDP problem as follows:

\begin{itemize}
  \item \underline{\textit{Belief States.}} The belief state of node $i$ is defined as the current probability distribution over $\bm{\mathcal{S}}_i$ on condition of the history so far.
Mathematically, the belief state of node $i$ in slot $t$ is denoted by 
\begin{equation}
  \bm{B}_{t,i} \triangleq \left[ B_{t,i}(\bm{s}_{t,i})\right]_{\bm{s}_{t,i}\in \bm{\mathcal{S}}_i}
\end{equation} 
with $\rVert \bm{B}_{t,i}\rVert_1=1 $, where $B_{t,i}(\bm{s}_{t,i})\triangleq \Pr\left( \bm{s}_{t,i}| \bm{h}_{t,i} \right)$ denotes the probability
assigned to state $\bm{s}_{t,i}$ with the current history $\bm{h}_{t,i}\triangleq\left\langle \bm{B}_{1,i},a_{1,i},\bm{o}_{1,i},a_{2,i},\dots,a_{t-1,i},\bm{o}_{t-1,i}\right\rangle$ of node $i$. 
As mentioned in Section \ref{SectionAoI}, ${D}_{t,i}$ is deterministic for a given history profile $\bm{h}_{t,i}$ since $\bm{h}_{t,i}$ includes $\bm{o}_{t-1,i}$. Therefore, $\bm{B}_{t,i}$ can also be represented by $\left\langle {D}_{t,i},\bm{b}_{t,i}\right\rangle$, where $\bm{b}_{t,i}\triangleq \left[ b_{t,i}(d_{t,i})\right]_{d_{t,i}\in \mathbb{Z}^+}$
denotes the belief state of the local age of node $i$, and $\rVert \bm{b}_{t,i}\rVert_1=1 $. 
Furthermore, $b_{t,i}\left( {d}_{t,i}\right)\triangleq \Pr\left(d_{t,i}|\bm{h}_{t,i}\right)$ denotes the probability assigned to ${d}_{t,i}$. 
Hence, we have $B_{t,i}\left( \bm{s}_{t,i}\right)= b_{t,i}\left( {d}_{t,i}\right) $ given~${D}_{t,i}$.

The network-wide belief state is defined as the current probability distribution over $\bm{\mathcal{S}}$ on condition of $\bm{h}_t\triangleq\left\langle \bm{B}_1,a_1,\bm{o}_1,\bm{a}_2,\dots,\bm{a}_{t-1},\bm{o}_{t-1}\right\rangle$, and it is also the state of the belief-MDP.
We denote the network-wide belief state in slot $t$ by 
\begin{equation}
  \bm{B}_t\triangleq\left[ B_t(\bm{s}_t)\right]_{\bm{s}_t\in \bm{\mathcal{S}}}=\left\langle \bm{D}_t,\bm{b}_t\right\rangle,
\end{equation}
where $\bm{b}_t\triangleq\left[ b_t(\bm{d}_t)\right]_{\bm{d}_t\in \bm{\mathcal{D}}}$ is the belief state of all local ages in slot $t$ with $b_t\left( \bm{d}_t\right)\triangleq \Pr\left( \bm{d}_t| \bm{h}_t \right) $ denoting the probability\footnotemark{} assigned to $\bm{d}_t$, and with $B_t(\bm{s}_t)\triangleq \Pr\left( \bm{s}_t| \bm{h}_t \right) $ denoting the probability assigned to $\bm{s}_t$.
\footnotetext{We omitted $\bm{h}_{t,i}$ in the definition of the belief state for concise notation.}
Thus, we have $\rVert\bm{B}_t\rVert_1=\rVert\bm{b}_t\rVert_1=1$.
With a given $\bm{D}_t$, the belief state of the local age of each node evolves independently in our POMDP framework, and thus we have ${B}_t\left( \bm{s}_t\right) =b_t\left( \bm{d}_t\right) =\prod_{i=1}^{N}{b}_{t,i}\left( d_{t,i}\right) $. 
Besides, we denote $\bm{\mathcal{B}}$ as the belief space, i.e., the collection of all possible $\bm{B}_t$.
$\bm{\mathcal{B}}$ is also called \textit{belief simplex} \cite{kaelbling1998planning}.

\item \underline{\textit{Belief Update.}} AP can update $\bm{B}_{t+1}$ from $\bm{B}_t$ at the end of slot $t$ after receiving new observations once the last action $\bm{a}_t$ is executed. Recall that $\bm{B}_t=\left\langle \bm{D}_t,\bm{b}_t\right\rangle$, both $\bm{D}_t$ and $\bm{b}_t$ need to be updated. Specifically, the destination AoI of node $i$ , i.e., the $i$-th component of $\bm{D}_t$, can be updated by
	\begin{equation}\label{Du}
		\begin{split}
			D_{t+1,i} =
			\begin{cases}
				D_{t,i}+1, & \text{if}\ \hat{d}_{t,i} = X,\\
				\hat{d}_{t,i}+1, & \text{otherwise}.
			\end{cases}
		\end{split}
	\end{equation}
	The update of $D_{t,i}$ is deterministic and independent from node to node.
	Moreover, $\bm{b}_{t+1}$ can be updated from $\bm{b}_t$ through the Bayes’ theorem as
	\begin{equation}\label{bu1}
		\begin{split}
			&b_{t+1}(\bm{d}_{t+1})\\
			&= \rho \sum_{\bm{d}_{t}\in\bm{\mathcal{D}}} b_{t}\left( \bm{d}_{t}\right) \Pr\left( \bm{d}_{t+1}|\bm{d}_{t} \right)\Pr\left( \hat{\bm{d}}_{t}|\bm{d}_{t},\bm{a}_{t}\right),
		\end{split} 
	\end{equation}
	where
	\begin{equation}\label{bu2}
		\rho=\\
		1/\sum_{\bm{d}_{t+1},\bm{d}_{t}\in\bm{\mathcal{D}}}b_{t}\left( \bm{d}_{t}\right) \Pr\left( \bm{d}_{t+1}|\bm{d}_{t} \right)\Pr\left( \hat{\bm{d}}_{t}|\bm{d}_{t},\bm{a}_{t}\right)
	\end{equation}
	is the Bayes normalizing factor. Considering the independent evolutions of $d_{t,i}$'s across nodes, we can also update $\bm{b}_{t}$ via updating $\bm{b}_{t,i}$ of each node $i$ individually. We omit the update equation of $\bm{b}_{t,i}$ here for brevity.

\item \underline{\textit{Actions.}} The action of the belief-MDP in slot $t$ is denoted by $\bm{a}_t\in \bm{\mathcal{A}}$, which is exactly same as that of the POMDP. 

\item \underline{\textit{Transition Function.}} The transition function of the belief-MDP is given by

\begin{equation}
	\begin{split}
		&\Pr\left(\bm{B}_{t+1}|\bm{B}_t,\bm{a}_t\right)\\
		&=\sum_{\bm{o}_t\in\bm{\mathcal{O}}}\Pr\left(\bm{B}_{t+1}|\bm{B}_t,\bm{a}_t,\bm{o}_t\right)\Pr\left(\bm{o}_t|\bm{B}_t,\bm{a}_t\right),
	\end{split}
\end{equation}
where
\begin{equation}
	\begin{split}
		&\Pr\left(\bm{o}_t|\bm{B}_t,\bm{a}_t\right)\\
		&=\sum_{\bm{s}_{t+1},\bm{s}_{t}\in\bm{\mathcal{S}}}B_{t}\left( \bm{s}_{t}\right) \Pr\left( \bm{s}_{t+1}|\bm{s}_{t},\bm{a}_t \right)\Pr\left( \bm{o}_{t}|\bm{s}_{t},\bm{a}_{t}\right),
	\end{split}
\end{equation}
and
\begin{equation}\label{BMDPTF}
\begin{split}
  & \Pr\left(\bm{B}_{t+1}|\bm{B}_t,\bm{a}_t,\bm{o}_t\right) \\ 
	& = \begin{cases}
		1,& \text{if the belief update with arguments} \\
		& \bm{B}_t,\bm{a}_t,\bm{o}_t \text{ returns}\ \bm{B}_{t+1},\\
		0,& \text{otherwise}.
	\end{cases}
\end{split}
\end{equation}

\item \underline{\textit{Policy.}} We adopt a deterministic stationary scheduling policy $\pi$ for the belief-MDP. The policy maps the belief space $\bm{\mathcal{B}}$ to the action space
in each slot.

\item \underline{\textit{Reward.}} Since the destination AoI is deterministic for the AP, the immediate expected reward on condition of belief state $\bm{B}_t$ is the same as that in the POMDP, i.e., $R\left( \bm{B}_t\right)\triangleq\mathbb{E}\left[r(\bm{s}_t)|\bm{B}_t\right]=\sum^N_{i=1} \omega_i D_{t,i} $. 
On this basis, the objective problem can be rewritten as
\begin{equation}
    \begin{split}
    \min_\pi\quad & \lim_{T\to \infty}\frac{1}{NT}\mathbb{E}\left[\sum^T_{t=1}R(\bm{B}_t)\Big|\bm{B}_1,\pi\right],\\
    \mbox{s.t.,}\quad &
    \rVert\bm{a}_t\rVert_1\le 1,\forall t,
    \end{split}
\end{equation}
where $\bm{B}_1$ is a predefined initial belief state. 
Recall that we assume $d_{0,i}=D_{0,i}=1,\forall i$, before running the network, and thus $\bm{B}_{1,i}=\left\langle 2,\left[ \lambda_i,1-\lambda_i,0,\cdots\right] \right\rangle,\forall i$.
\end{itemize}

We remark that
the belief update is computationally complicated when the dimension of the belief states is high, and is impractical when the dimension goes to infinity.
Moreover, the continuousness of the belief space $\bm{\mathcal{B}}$ leads to a PSPACE hardness of optimizing the EWSAoI of the belief-MDP optimally \cite{papadimitriou1987complexity}.
Thus, it is intractable to optimize the EWSAoI of the network exactly.
As such, we are motivated to further analyze the belief-MDP to find a more feasible solution.

\subsection{Belief-MDP Simplification}

We subsequently show the existence of a simplified representation of the belief-MDP with the given $\bm{B}_1$. 
To start, we have the following definition:

\begin{definition}\label{DefineKM}
	Assume AP schedules node $i$ in slot $t$ with observation $d_{t,i}=k_i$, and then does not receive
	any packet from node $i$ in the following $m_i$ slots. Define the local age belief state of node $i$ in slot $t+m_i$ by $\bm{c}(k_i,m_i)$, namely, the belief of node $i$ with the last observation $k_i$
	followed by $m_i$ elapsed slots.
\end{definition}

For convenience, we ignore index $i$ for nodes and introduce the following proposition.

\begin{proposition}\label{prop_km}
	The distribution vector of the local age belief state $\bm{c}(k,m)$ of node $i$ in slot $t$ can be given by
	\begin{equation}\label{ckm}
		\begin{split}
			\bm{c}(k,m)&=\left[ c_{k,m}\left( d_t\right) \right]_{d_t\in\mathbb{Z}^+}\\
			&=\left[ \lambda,\lambda\gamma,\lambda\gamma^2,
		\cdots,\lambda\gamma^{m-1},0,\cdots,0,\gamma^m,0,\cdots\right],
		\end{split} 
	\end{equation}
	where $k,m \in \mathbb{Z}^+ $, $\gamma=1-\lambda$, and $c_{k,m}\left( {d}_t\right) $ denotes the probability assigned to ${d}_t$. The position of entry $\gamma^m$ is $k+m$, and this denotes that the corresponding destination AoI of entry $\gamma^m$ is $k+m$.
\end{proposition}
\begin{proof}
	See Appendix \ref{appA}.
\end{proof}

Notice that \textbf{Proposition \ref{prop_km}} follows the evolution branch of local age belief state in \cite[\textbf{Proposition} 4]{9590543}. However, the belief state in \cite{9590543} is the distribution of the local age only, while that in this paper also involves the destination AoI. Moreover, a truncation was given to the local age in \cite{9590543} but not used in this paper.


Define a group of belief states that have the AoI equal to $k+m$ together with $\bm{c}(k,m)$ defined in \textbf{Proposition \ref{prop_km}} as $\bm{C}(k,m)\triangleq \left\langle k+m,\bm{c}(k,m)\right\rangle $ for
$m, k \in \mathbb{Z}^+$.
Denote by $\bm{\mathcal{C}}$ the collection of all possible $\bm{C}(k,m)$. Then, we have the following corollary.

\begin{corollary}\label{corokm}
	Suppose the network has a certain belief state, i.e., $\bm{b}_{0,i}=\bm{e}_1,D_{0,i}=1,\forall i$ before running, then $\bm{B}_{t,i}\in \bm{\mathcal{C}}$ for $t=1,2,\cdots,T,\forall i$. 
\end{corollary}

\begin{proof}
	We use induction to prove it. First, it is clear that $\bm{b}_{1,i}=\bm{c}(1,1),D_{1,i}=2$, and hence $\bm{B}_{1,i}\in \bm{\mathcal{C}}$. Suppose $\bm{B}_{t,i}=\left\langle k_{t,i}+m_{t,i},\bm{c}(k_{t,i},m_{t,i}) \right\rangle \in\bm{\mathcal{C}},\forall i$, where $k_{t,i}$ and $m_{t,i}$ denote the last observation of local age and the number of slots elapsed since the last observation of node $i$ in slot $t$, respectively. Then, if node $i$ is scheduled and the status update is successfully received by the AP, we have $\bm{b}_{t+1,i}=\bm{c}(\hat{k}_{t,i},1)$ and $D_{t+1,i}=\hat{k}_{t,i}+1$, where $\hat{k}_{t,i}$ is the local age observation of node $i$ in slot $t$ and $\hat{k}_{t,i}\in\left\lbrace 1,2,\cdots,m_{t,i}\right\rbrace \cup \left\lbrace k_{t,i}+m_{t,i} \right\rbrace$. This means that $\bm{B}_{t+1,i}\in\bm{\mathcal{C}},\forall i$. If node $i$ is not scheduled or the transmission fails, $\bm{B}_{t+1,i}=\left\langle k_{t,i}+m_{t,i}+1,\bm{c}(k_{t,i},m_{t,i}+1) \right\rangle \in \bm{\mathcal{C}},\forall i$. This completes the proof.
	
\end{proof}

Based on \textbf{\textit{Corollary} \ref{corokm}}, each infinite dimensional belief state $\bm{B}_{t,i}\in\bm{\mathcal{C}}$ can be sufficiently represented by two positive integers $k_{t,i}$ and $m_{t,i}$, with $D_{t,i}=k_{t,i}+m_{t,i}$ and $\bm{b}_{t,i}=\bm{c}(k_{t,i},m_{t,i})$. Hence, the belief-MDP framework in \textbf{Section \ref{POMDP}} can be characterized in a much simpler form.
We name this simplified representation of belief MDP as Last-Observation-Characterized (LOC) belief-MDP.
The actions of the LOC belief-MDP are the same as that of the original belief-MDP. The other components of the LOC belief-MDP are presented as follows.

\begin{itemize}
	\item \underline{\textit{States.}} The state of node $i$ in slot $t$ is denoted by $\bm{z}_{t,i}\triangleq \left[ k_{t,i},m_{t,i}\right] $, where $k_{t,i},m_{t,i} \in \mathbb{Z}^+$ are defined in \textbf{\textit{Corollary} \ref{corokm}}. Then, the network-wide state in slot $t$ is denoted by $\bm{z}_t \triangleq \left[ \bm{k}_t,\bm{m}_t\right] $, where $\bm{k}_t\triangleq \left[ k_{t,1},k_{t,2},\dots,k_{t,N}\right] \in \bm{\mathcal{D}}$ and $\bm{m}_t\triangleq \left[ m_{t,1},m_{t,2},\dots,m_{t,N}\right] \in \bm{\mathcal{D}}$. Define $\bm{\mathcal{Z}}\triangleq\bm{\mathcal{D}}\times\bm{\mathcal{D}}$ as the space set of $\bm{z}_t$. $\bm{\mathcal{Z}}$ also corresponds to the feasible part of the belief space $\bm{\mathcal{B}}$ for belief states with the initialization in \textbf{\textit{Corollary} \ref{corokm}}.
	
	\item \underline{\textit{Transition Function.}} We define the transition function of the LOC belief-MDP as $\Pr\left(\bm{z}_{t+1}|\bm{z}_t,\bm{a}_t\right)$, which is given~by
	\begin{equation}
		\label{eq_Dkm}
		\Pr\left( \bm{z}_{t+1}|\bm{z}_t,\bm{a}_t \right) =\prod_{i=1}^{N}\Pr\left(\bm{z}_{t+1,i}|\bm{z}_{t,i},a_{t,i} \right), 
	\end{equation}
	where
	\begin{equation}\label{kmTrans}
		\begin{split}
			&\Pr\left( \bm{z}_{t+1,i}|\bm{z}_{t,i},a_{t,i}=1 \right)\\
		&=
		\begin{cases}
			p_i\lambda_i(1-\lambda_i)^{d-1},& \text{if}\ k_{t+1,i}=d \text{ and}\\
			& m_{t+1,i}=1, \\
			p_i(1-\lambda_i)^{m_{t,i}},& \text{if}\ k_{t+1,i}=k_{t,i}+m_{t,i}\\
			& \text{and } m_{t+1,i}=1,\\
			1-p_i,& \text{if}\ k_{t+1,i}=k_{t,i} \text{ and}\\
			& m_{t+1,i}=m_{t,i}+1,\\
			0,& \text{otherwise},
		\end{cases}
		\end{split}
	\end{equation}
	with $d\in\left\lbrace
	1,2,\cdots,m_{t,i}\right\rbrace$.
	Furthermore, $\Pr\left( \bm{z}_{t+1,i}|\bm{z}_{t,i},a_{t,i}=0 \right)=1$ if $k_{t+1,i}=k_{t,i}$ and $m_{t+1,i}=m_{t,i}+1$.

	\item \underline{\textit{Reward.}} The expected immediate reward given a state $\bm{z}_t$ is rewritten as $R\left( \bm{z}_t\right)\triangleq \sum^N_{i=1} \omega_i (k_{t,i}+m_{t,i}) $.
	
	\item \underline{\textit{Policy.}} The policy for the LOC belief-MDP framework is the same as that in \textbf{Section \ref{POMDP}} with a different domain $\bm{\mathcal{Z}}$. It can be equivalently denoted by $\pi:\bm{\mathcal{Z}}\longmapsto \bm{\mathcal{A}}$.
	
\end{itemize}


In typical work on solving a belief-MDP, one need to use the \textit{Backup operation} \cite{kaelbling1998planning,zhang2001speeding} to repeatedly find more feasible belief states and update the feasible belief space horizon by horizon.
It is computationally complicated, and unlikely to reach most of feasible belief states in the belief simplex.
However, with the above simplification, we reduce the space of belief states sharply from the continuous space $\bm{\mathcal{B}}$ to a discrete space $\bm{\mathcal{Z}}$.
That enables us to directly obtain the full feasible space of the belief states without using the inefficient \textit{Backup operation}.
Furthermore, the completed transition probabilities of belief states can be obtained by \eqref{eq_Dkm}.
\begin{figure}[t]
	\centering
	\includegraphics[width=0.52\textwidth]{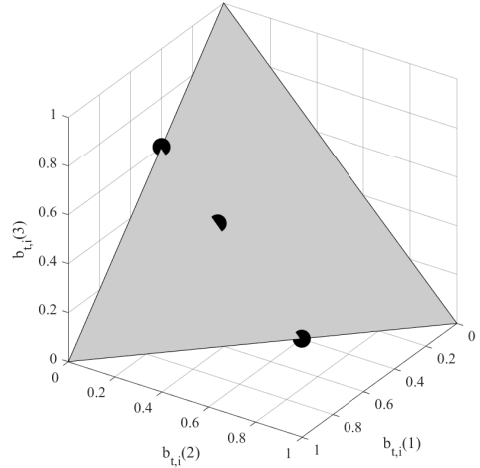}
	\caption{The sub-region of $\bm{\mathcal{B}}$ and its reduced feasible space, i.e., points $[0.4,0.6,0,\cdots]$,$[0.4,0,0.6,0,\cdots]$, and $[0.4,0.24,0.36,0,\cdots]$ in the three-dimensional space with $N=1$ and $\lambda=0.4$.}
	\label{Simplex}
\end{figure}
Fig.\ref{Simplex} illustrates one example of the space reduction, where we have one node with its status update arrival rate $\lambda=0.4$. The gray triangle plane is the sub-region of $\bm{\mathcal{B}}$ in the three-dimensional space, on which each point is associated with a possible local age belief state. After the simplification, the sub-region of the belief space $\bm{\mathcal{B}}$ can be reduced to three feasible belief states, i.e., the three points plotted on the sub-region.

\begin{remark}\label{RemarkM}
We can extend the above LOC belief-MDP simplification process to the scenario with Markovian arrival processes.
Specifically, the belief states of a node can still be characterized by two-dimensional vectors.
More details can be found in Appendix D.
\end{remark}

\section{POMW Policy}


Based on the LOC belief-MDP, we can use the conventional DP approach to solve the AoI scheduling problem.
However, the LOC belief-MDP is formulated for a multiuser model, thus the DP would suffer from the curse of the dimensionality as the number of end devices increases.
To circumvent such a problem, we propose a low-complexity policy for the EWSAoI optimization in the considered network with partial observations, named POMW policy.

We remark that a downlink network with the same status update traffic model as ours was investigated in \cite{8933047}.
Different from our network, the local age of the status update packets are fully observable at the AP due to the downlink setting.
The authors devised an Age-based Max-Weight policy by leveraging the Lyapunov Optimization \cite{neely2010stochastic}.
This policy minimizes a defined Lyapunov drift on condition of the fully observable local age and destination AoI in each slot.
Hereafter, we call it Fully Observable Max-Weight (FOMW) policy.
Moreover, for brevity, we use \textit{``FON''} to represent the network with the fully observable states in \cite{8933047} and \textit{``PON''} to represent our considered network with partial observations in the rest of this paper.

Inspired by \cite{8933047}, we apply the Lyapunov Optimization to devise a low-complexity policy, i.e., the POMW policy, which extends the FOMW policy developed in \cite{8933047}. 
To that end, we will define a Lyapunov Function based on the EWSAoI of the network. The POMW policy attempts to minimize the expected drift of the Lyapunov Function under condition of the current belief state and destination AoI in each slot $t$. Therefore, the EWSAoI of the network can be optimized with lower computational complexity.


We define the linear Lyapunov Function as
\begin{equation}
\label{before_LD}
  L(t)=\frac{1}{N}\sum^N_{i=1}\beta_i D_{t,i},
\end{equation}
where $\beta_i>0$ is an hyper-parameter that can be used to tune the POMW policy to different network configurations. 
The Lyapunov Drift is defined as
\begin{equation}\label{LD}
  \Delta\left[\bm{B}_t\right]=\mathbb{E}\left[L(t+1)-L(t)|\bm{B}_t\right].
\end{equation}
The Lyapunov Drift $\Delta\left[\bm{B}_t\right]$ refers to the expected increase of the Lyapunov Function $L(t)$ in one slot. Hence, by minimizing the drift in \eqref{LD}, the POMW policy equivalently reduces $L(t)$. Consequently, the EWSAoI of the network is kept low.

To develop the POMW policy for the Lyapunov Drift minimization,
we analyze the expression for the drift in \eqref{LD}. 
Recall the definition of $\bm{B}_t$, we realize that the value of $L(t)$ is fixed with a given $\bm{B}_t$. Thus minimizing the Lyapunov Drift in \eqref{LD} is equivalent to minimizing $\mathbb{E}[ L(t + 1)|\bm{B}_t]$.
Recall the evolution of $D_{t,i}$ given in \eqref{eqAoI}, and we have
\begin{equation}\label{EWSAoINS}
	\begin{split}
	  & \mathbb{E}[ L(t + 1)|\bm{B}_t] \\
		= & \sum_{i=1}^{N} \dfrac{1}{N} \mathbb{E}\left[ \beta_i D_{t+1,i}\Bigg|\bm{B}_t\right]\\
		= & \sum_{i=1}^{N} \dfrac{\beta_i}{N}\left[p_i a_{t,i} \sum_{d\in\mathbb{Z}^+}b_{t,i}(d)(d+1)+(1-p_i a_{t,i})(D_{t,i}+1) \right]\\
		= & \dfrac{1}{N} \left[ - \sum_{i=1}^{N} a_{t,i}\beta_i p_i G_{t,i} +\sum_{i=1}^{N}\beta_i(D_{t,i}+1) \right],
	\end{split} 
\end{equation}
where
\begin{equation}\label{G}
  G_{t,i} : = D_{t,i}-\sum_{d\in\mathbb{Z}^+}b_{t,i}(d)d.
\end{equation}

Eq.  \eqref{EWSAoINS} leads to following proposition:

\begin{proposition}\label{ProPOMW}
To minimize the Lyapunov Drift in slot $t$, the POMW policy should schedule node $i$ with the maximal $\beta_i p_i G_{t,i}$.
\end{proposition}
The proof of \textbf{Proposition \ref{ProPOMW}} is straightforward and hence is omitted. Before the POMW policy making the scheduling decision, the belief probabilities $b_{t,i}(d)$ need to be updated based on the observations of the previous slot.

\begin{remark}
Note that when local age is fully observed, we have
\begin{equation}
    \sum_{d\in\mathbb{Z}^+}b_{t,i}(d)d= d_{t,i},\forall t,i.
\end{equation}
In this case, the POMW policy will schedule node $i$ with the maximal $\beta_ip_i\left(D_{t,i}-d_{t,i}\right)$
in each slot, which exactly coincides with the criterion of the FOMW policy presented in \cite{8933047}.
This observation indicates that the POMW policy is a generalization of the FOMW policy.
\end{remark}

However, it is hard to implement this online policy on the fly due to the high computational complexity. 
In each slot, the POMW policy selects an action $\bm{a}_t$ by minimizing \eqref{EWSAoINS}. 
This step requires $O\left(N|\bm{\mathcal{A}}| |\mathbb{Z}^+|\right)$ operations. Subsequently, the policy updates the local age belief states for the next slot by the Bayes' theorem. 
Such an update step requires $O\left(N|\mathbb{Z}^+|^2\right)$ operations.
Those two steps are computationally intractable since $\mathbb{Z}^+$ is an infinite set\footnotemark{}.
Thus, the straightforward application of the FOMW policy to our problem could be impractical. 
\footnotetext{One can truncate the maximum value of AoI to make the computation feasible. However, a sufficiently large cap of the AoI should be applied to ensure the accuracy of the truncation, which still leads to unacceptably high computational complexity.}

Thanks to the LOC belief-MDP framework proposed in \textbf{Proposition \ref{prop_km}}, we are able to simplify the expression of $G_{t,i}$ from complex expectation calculation to a closed-form expression of only three parameters. More specifically, it can be expressed as
\begin{equation}
	\begin{split}
		G_{t,i} & = k_{t,i}+m_{t,i}-\bm{c}(k_{t,i},m_{t,i})\bm{n}\\
		&=m_{t,i}+\left[ 1-\left( 1-\lambda_i\right)^{m_{t,i}} \right]\left( k_{t,i}-\frac{1}{\lambda_i}\right),
	\end{split}
\end{equation}
where $\bm{n}=\left[ 1,2,3,\cdots\right]^T $. 
By now, we can formally describe the POMW policy in \textbf{Algorithm \ref{AlgMW}}. The POMW policy can minimize the Lyapunov Drift with low computational complexity, and consequently optimize the EWSAoI of the network.
\begin{algorithm}
	\SetAlgoLined
	
	Initialization: $t=1,m_{t,i}=1,k_{t,i}=1,\forall i$\footnotemark{}\;
	\While{each new slot $t$}{
		\For{each node $i$}{
			$G_{t,i}=m_{t,i}+\left[ 1-\left( 1-\lambda_i\right)^{m_{t,i}} \right]\left( k_{t,i}-\frac{1}{\lambda_i}\right)$\;
		}
		Schedule node $j$ in the current slot, where\\
		$j=\arg \max\limits_{i} \beta_ip_iG_{t,i}$\;
		Obtain the local age observation $\hat{d}_{t,j}$ of node $j$\;
		\uIf{$\hat{d}_{t,j}=X$}{
			$m_{t+1,j}=m_{t,j}+1$\;
		}
		\ElseIf{$\hat{d}_{t,j}=d\in\mathbb{Z}^+$}{
			$m_{t+1,j}=1,k_{t+1,j}=d$\;
		}
		$m_{t+1,i}=m_{t,i}+1,\forall i\ne j$\;
		
		$t=t+1$;
	}
	\caption{POMW Policy}\label{AlgMW}
\end{algorithm}
\footnotetext{To ease understanding and simplify expressions, we set such an initialization. Without loss of generality, we can also select any $\bm{b}_1\in\bm{\mathcal{B}}$ for the initialization. In that case, when $m_{t,i}$ and $k_{t,i}$ do not exist for some $i,t$, we can update belief states by \eqref{bu1} and calculate $G_{t,i}$ by \eqref{G}.}
Note that in \textbf{Algorithm \ref{AlgMW}}, the updates of $k_{t,i}$ and $m_{t,i}$ are based on the transition function of the LOC belief-MDP given in \eqref{kmTrans}.

Thanks to the proposed simplification, the complexity of the step to select an action is reduced from 
$O\left(N|\bm{\mathcal{A}}| |\mathbb{Z}^+|\right)$ to $O\left(N |\bm{\mathcal{A}}|\right)$. The complexity of updating states is reduced from $O\left(N|\mathbb{Z}^+|^2\right)$ to $O\left(2N\right)$. Moreover, we do not need to set a truncation on destination AoI or local age when implementing \textbf{Algorithm \ref{AlgMW}}.



\section{Performance Analyses}


In this section, we first introduce a low-complexity policy named Randomized Scheduling (RS) policy and analyze its EWSAoI performance. Based on its performance, we derive the upper bounds for the EWSAoI performance of the POMW policy. We also analyze the performance guarantee of the POMW policy by comparing its EWSAoI performance with a universal lower bound in the literature.

\subsection{RS Policy}

In \cite{8933047}, an RS policy was proposed to optimize the AoI of the network. In the RS policy, node $i$ is scheduled
with probability $\mu_i\in\left( 0,1\right] $ in each slot. The scheduling probabilities are time-invariant and satisfy $\sum_{i=1}^{N}\mu_i\le1 $.
Notice that the actions of the RS policy is independent of the network states, thus this policy can also be adopted in the considered PON.
With reference to the proof of \cite[Prop. 4]{8933047}, we give the EWSAoI performance of the RS policy, denoted by $R^{RS}$, in \textbf{\textit{Proposition} \ref{prop1}}. 

\begin{proposition}\label{prop1}
	The EWSAoI of the network under the RS policy with scheduling probabilities $\left\lbrace \mu_i\right \rbrace^N_{i=1}$~is
	\begin{equation}\label{prop1Eq}
		R^{RS}=\frac{1}{N}\sum_{i=1}^{N}\omega_i\left(\dfrac{1}{\lambda_i}+\dfrac{1}{p_i\mu_i} \right). 
	\end{equation}
\end{proposition}

Note that the EWSAoI in \eqref{prop1Eq} is slightly different from that in \cite{8933047} due to the difference in the local age evolution of two systems.
Denote by $\left\lbrace \mu^*_i\right \rbrace^N_{i=1} $ the optimal scheduling probabilities of all node, the optimal RS policy is given as follows \cite[Th. 5]{8933047}.
\begin{theorem}
	Consider the network under the RS policy. The optimal scheduling probabilities are 
	\begin{equation}
		\mu^*_i=\dfrac{\sqrt{\omega_i/p_i} }{\sum_{j=1}^{N}\sqrt {\omega_j/p_j}},
	\end{equation}
	and correspondingly,
	\begin{equation}
		R^{{RS}^*}=\frac{1}{N}\left[ \sum_{i-1}^{N}\dfrac{\omega_i}{\lambda_i}+\left( \sum_{i-1}^{N}\sqrt{\dfrac{\omega_i}{p_i}}\right)^2 \right]. 
	\end{equation}
\end{theorem}

According to \cite[Th.10]{8933047}, $R^{RS*}$ is the upper bound of the EWSAoI performance of the FOMW policy, denoted by $R^{FOMW}$.
The FOMW policy can be regarded as the full-observed counterpart of the POMW policy.

\subsection{Upper Bounds of the POMW Policy}
Built upon the proposed LOC belief-MDP, we now derive two upper bounds for the POMW policy. 
One of them is the EWSAoI performance of a particular RS policy, as stated in the following theorem:
\begin{theorem}\label{UBTh}
	The EWSAoI performance of the POMW policy with $\beta_i=\omega_i/\lambda_i\mu'_ip_i,\forall i$, denoted by $R^{POMW}$, is upper bounded by
	\begin{equation}\label{UBPOMW}
		R^{POMW}\le\frac{1}{N}\sum_{i=1}^{N}\omega_i\left( \dfrac{1}{\lambda_i\mu'_ip_i}+1\right) \le R^{RSM},
	\end{equation}
	where
	\begin{equation}
		\mu'_i=\dfrac{\sqrt{\omega_i/\lambda_ip_i}}{\sum_{j=1}^{N}\sqrt{\omega_j/\lambda_jp_j}}, \forall
		i,
	\end{equation}
	are a series of scheduling probabilities of all nodes in the network. $R^{RSM}$ is the EWSAoI of an RS policy with the corresponding scheduling probabilities
	\begin{equation}
		\mu^M_i=\dfrac{\sqrt{\omega_i\lambda_i/p_i}}{\sum_{j=1}^{N}\sqrt{\omega_j/\lambda_j p_j}}, \forall
		i.
	\end{equation}
\end{theorem}

\begin{proof}
	We prove it by leveraging the introduced RS policy. See Appendix \ref{appD} for details.
\end{proof}
Note that the value assigned to $\beta_i$, which depends on $\mu'_i$, can be attained by minimizing the upper bound of the EWSAoI, given by
$\frac{1}{N}\sum_{i=1}^{N}\omega_i\left( 1/(\lambda_i\mu'_ip_i)+1\right)$.
A similar method was used in \cite{neely2010stochastic}.

\subsection{Performance Guarantee of the POMW Policy}

Based on \textbf{\textit{Theorem} \ref{UBTh}}, we can analyze the performance guarantee of the POMW policy theoretically. Firstly, we introduce a universal lower bound given in \cite[Th.3]{8933047}.
This lower bound applies to any feasible scheduling policy, and is applicable to both FONs and PONs. 
The universal lower bound of the EWSAoI performance of any policies is given~by
\begin{alignat}{2}
	L_B=\min\limits_{\left\lbrace q_i\right\rbrace^N_{i=1} }\quad & 
	\frac{1}{2N}\sum_{i=1}^{N}\omega_i\left( \frac{1}{q_i}+3\right),\label{LB1}\\
	\mbox{s.t.,}\quad & \sum_{i=1}^{N}q_i/p_i\le 1, \label{LB2}\\
	& q_i\le\lambda_i,\forall i.\label{LB3}
\end{alignat}	
The solution $q^*_i$ of the above problem can be obtained following \cite[Algorithm 1]{8933047}, and the lower bound is
\begin{equation}
  L_B= \frac{1}{2N}\sum_{i=1}^{N}\omega_i\left( \frac{1}{q^*_i}+3\right).
\end{equation}

\begin{corollary}\label{RLR}
	The performance of the POMW policy with $\beta_i=\omega_i/\lambda_iq^*_i$ follows that 
	\begin{equation}
		\dfrac{R^{POMW}}{L_B}<\dfrac{2}{\lambda_{min}},
	\end{equation}
	where $\lambda_{min}\triangleq \min\left\lbrace\lambda_i \right\rbrace^N_{i=1} $. 
\end{corollary}

\begin{proof}
	See Appendix \ref{appE}
\end{proof}

\begin{remark}
We use the ratio between $R^{POMW}$ and $L_B$ to evaluate the performance guarantee of the POMW policy. \textbf{Corollary \ref{RLR}} indicates that the ratio is inversely proportional to the packet arrival rates of nodes $\lambda_i$ in the network. When the network is close to the ``generate-at-will'', i.e., $\lambda_i\to 1,\forall i$, the ratio of $R^{POMW}$ and $L_B$ with $\beta_i=\omega_i/\lambda_iq^*_i$ is smaller than $2$. This coincides with the performance guarantee of
a counterpart FOMW policy devised for the ``generate-at-will'' system in \cite{sun2019age}. 
\end{remark}

\section{Numerical Results}

In this section, we first compare the proposed POMW policy with its fully observable counterparts. We then verify the theoretical analyses on the POMW policy.
Finally, we compare the performance of the POMW policy with that of two baseline policies in PONs.

\subsection{Comparisons with Fully Observable Counterparts}

The EWSAoI performances of the POMW and FOMW policies are obtained via $2000$ Monte-Carlo simulation runs. 
$\beta_i$ is set as $\omega_i/\lambda_i\mu'_ip_i$ for both of the two policy.
The AoI performance of the RS policies, the universal lower bound $L_B$, and the upper bound of $R^{POMW}$ are computed using \eqref{prop1Eq}, \eqref{LB1}, and \eqref{UBMW}, respectively. 
 
\begin{figure}[t]
	\centering
	\includegraphics[width=0.48\textwidth]{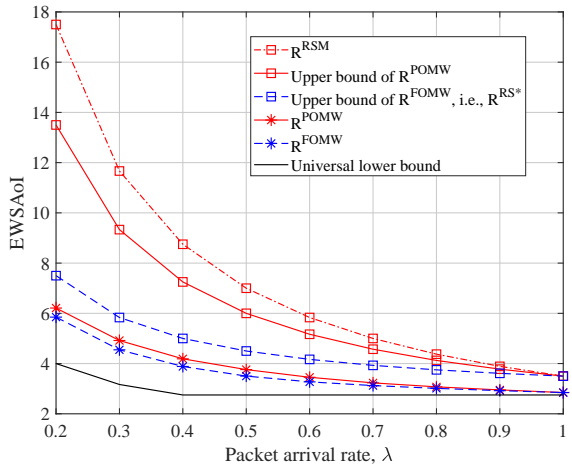}
	\caption{EWSAoI performance versus an increasing packet arrival rate, where $N=2$, $p_1=p_2=0.8$, and $\omega_1=\omega_2=1$.}
	\label{AoI_V_La}
\end{figure}

In Fig. \ref{AoI_V_La}, we illustrate the EWSAoI of
the POMW policy, its corresponding upper bounds, and the universal lower bound $L_B$ with increasing packet arrival rate. The $R^{FOMW}$ in the FON and its corresponding upper bound, i.e., the optimal RS policy, are given as benchmarks. We set $N=2$, $\omega_1=\omega_2=1$, $p_1=p_2=0.8$, $\lambda_1=\lambda_2=\lambda$, $D=20$, and $T=100$. Fig. \ref{AoI_V_La} shows that all curves decrease as $\lambda$ increases.
This is intuitive because the EWSAoI decreases when the status update packets arrive at nodes more frequently.
Furthermore, the value of the universal lower bound is the smallest, the value of $R^{RSM}$ is the largest, and $R^{POMW}$ and $R^{FOMW}$ are lower than their corresponding upper bounds, respectively. These relationships
validate the analysis given in the previous section.
Fig. \ref{AoI_V_La} also shows that $R^{POMW}$ is larger than $R^{FOMW}$, and the upper bound of $R^{POMW}$ is larger than that of $R^{FOMW}$.
This is intuitive because the POMW policy in the PON only knows the packet arrival rate and some occasional observations,
while
the FOMW policy in
the FON can utilize the
fully observed state information.
Furthermore, the gap between the $R^{FOMW}$ and its upper bound decreases slowly as $\lambda$ increases, while the gap between the $R^{POMW}$ and its upper bound decreases much quickly.
This phenomenon can be explained by refering to the expressions of the upper bounds, given by \eqref{prop1Eq} and \eqref{UBPOMW}. 
It is obvious that $1/\lambda_i$ in \eqref{UBPOMW} always has larger coefficient than that of \eqref{prop1Eq}.
Hence, the upper bound of $R^{POMW}$ increases faster than that of $R^{FOMW}$ with the decrease of $\lambda$.
Moreover, the EWSAoI performances of the POMW and FOMW polices and their corresponding upper bounds converge to the same value when $\lambda = 1$. This is because both the PON and FON approach to the ``generate-at-will'' model when $\lambda$ tends to $1$.

\begin{figure}[t]
	\centering
	\includegraphics[width=0.48\textwidth]{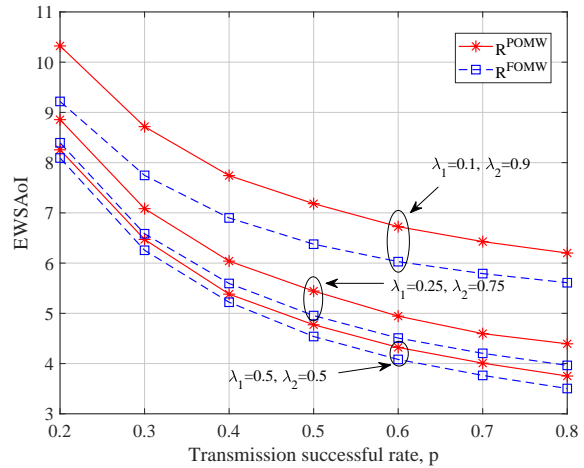}
	\caption{$R^{POMW}$ and $R^{FOMW}$ versus transmission successful rate with different combinations of $\lambda_i$ in PON and FON, where $N=2$, and $\omega_1=\omega_2=1$.}
	\label{AoI_V_DL}
\end{figure}

In Fig. \ref{AoI_V_DL}, we compare $R^{FOMW}$ and $R^{POMW}$ versus transmission successful rate $p$ with three pairs of packet arrival rates $\{ \lambda_1, \lambda_2 \}$.
We set $p_1=p_2=p$, $\omega_1=\omega_2=1$, and three pairs of packet arrival rates $\lambda_1=\lambda_2=0.5$; $\lambda_1=0.25$, $\lambda_2=0.75$; and $\lambda_1=0.1$, $\lambda_2=0.9$.
Fig. \ref{AoI_V_DL} shows that the gap between $R^{FOMW}$ and $R^{POMW}$ becomes larger when the gap between two $\lambda_i$s increase. 
This can be explained by combining \textbf{\textit{Corollary} \ref{RLR}} and \cite[Th.5]{8933047}. \textbf{\textit{Corollary} \ref{RLR}} indicates that the performance guarantee of $R^{POMW}$ is inversely proportional to $\lambda_{min}$, and \cite[Th.5]{8933047} indicates that the performance guarantee of the fully observed counterpart is a constant. A larger arrival rate gap results in a smaller $\lambda_{min}$, and consequently a larger gap between the $R^{POMW}$ and $R^{FOMW}$. 
Fig. \ref{AoI_V_DL} also shows that
$R^{POMW}$ and $R^{FOMW}$ decreases in all cases when the packet transmission rate $p$ increases. This is because the destination AoI $D_{t,i}$ drops to the local age $d_{t,i}+1$ more frequently with larger $p$. 



\begin{figure}[t]
	\centering
	\includegraphics[width=0.48\textwidth]{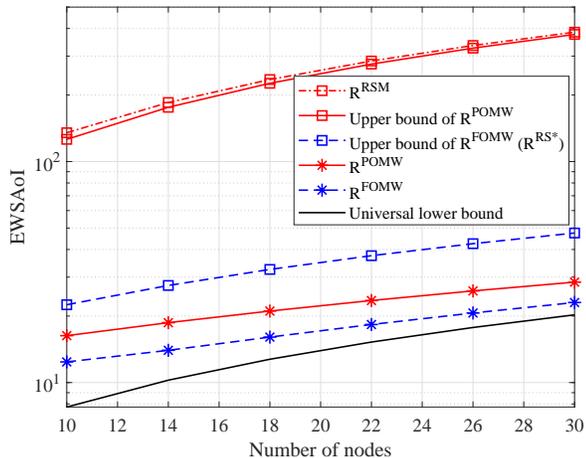}
	\caption{EWSAoI performance as the number of nodes increases, where $\lambda_1=\lambda_2=0.1$, $p_1=p_2=0.8$, and $\omega_1=\omega_2=1$.}
	\label{AoI_V_K}
\end{figure}

\begin{figure}[t]
	\centering
	\includegraphics[width=0.48\textwidth]{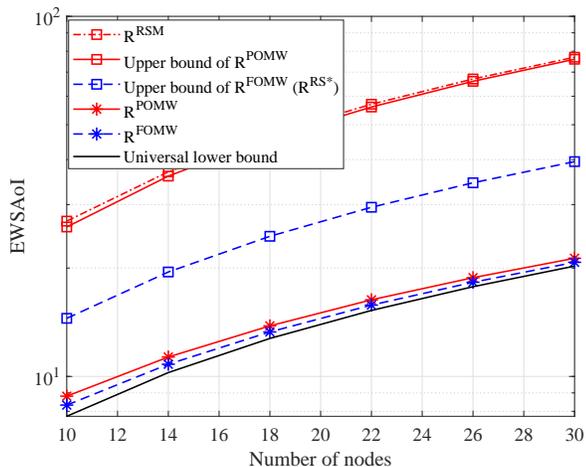}
	\caption{EWSAoI performance as the number of nodes increases, where $\lambda_1=\lambda_2=0.5$, $p_1=p_2=0.8$, and $\omega_1=\omega_2=1$.}
	\label{AoI_V_K1}
\end{figure}

In Fig. \ref{AoI_V_K}, we depict the EWSAoI of the POMW and FOMW policies with increased number of nodes, and compare them with corresponding upper bounds. We also include $R^{RSM}$ and the universal lower bound as benchmarks.
We set $\lambda_i=0.1$ and $p_i=0.8$ for all nodes, and increase the number of nodes from $10$ to $30$.
It is shown in Fig. \ref{AoI_V_K} that all curves increase as number of nodes increases.
This is because each node has fewer chances to be scheduled when the number of nodes increases, thus its AoI has fewer chances to decrease.
Fig. \ref{AoI_V_K1} plots the same set of curves as in Fig. \ref{AoI_V_K}, where the values of $\lambda_1$ and $\lambda_2$ are both set to be $0.5$.
Fig. \ref{AoI_V_K1} shows similar phenomenon observed in Fig. \ref{AoI_V_K}. 
Furthermore, the performance of all schemes improves from Fig. \ref{AoI_V_K} to Fig. \ref{AoI_V_K1}, which is expected since the packet arrival rates are increased.

\subsection{Comparison with Baseline Policies}
In the following, we show the advantages of the proposed POMW policy in PONs over two
baseline policies.
The baseline policies are described as follows:
\begin{itemize}
  \item [1)] \textbf{\textit{Round Robin (RR) policy:}} In the RR policy, nodes are scheduled by the AP in a circular order to ensure a fair scheduling opportunity among the nodes.
  \item [2)] \textbf{\textit{Max weighted AoI (MWA) policy:}} The MWA policy does not need the knowledge of nodes' local age. 
  Specifically, the AP always schedules the node $j$ with 
  $j=\arg \max\limits_{i}\omega_ip_iD_{t,i}$.
\end{itemize}
In the following figures, the EWSAoI performance of all policies are obtained via $10000$ Monte-Carlo simulation runs.

\begin{figure}[t]
	\centering
	\includegraphics[width=0.48\textwidth]{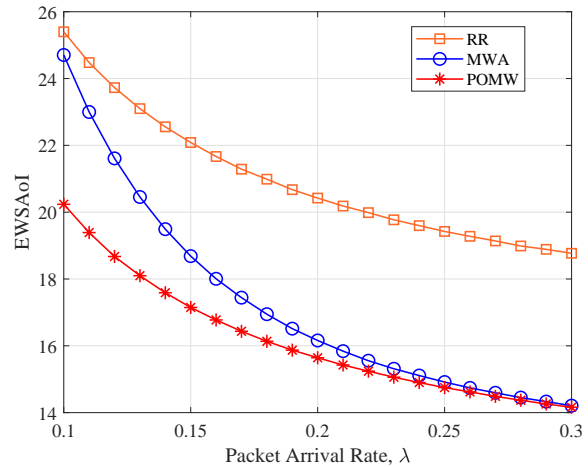}
	\caption{EWSAoI performance versus packet arrival rate $\lambda$ in the PON with $N=10$, and $\omega_i=1$, $p_i=0.5$, $\forall i$.}
	\label{Baseline_B}
\end{figure}

Fig. \ref{Baseline_B} shows the EWSAoI performance of the POMW policy, the MWA policy, and the RR policy in a symmetric PON. The weight coefficients $\omega_i$ and transmission successful rates $p_i$ of all nodes are the same.
We can observe from Fig. \ref{Baseline_B} that
the RR policy has the worst performance.
This is intuitive because the RR policy does not use the observations of the network states, while the MWA and POMW policies make decisions depending on the observations.
Moreover, 
the POMW policy is superior to the MWA policy when $\lambda$ is small, but the EWSAoI performances of these two policies tend to coincide when $\lambda \ge 0.25$.
This is owing to the fact that the MWA policy only leverages the observations of the destination AoI, while the POMW policy uses the observations of both the destination AoI and local age.
Furthermore, in a symmetric PON, the expected local ages of all nodes tend to be symmetric as $\lambda$ grows. 
In this case, the POMW policy and the MWA policy becomes equivalent.

\begin{figure}[t]
	\centering
	\includegraphics[width=0.48\textwidth]{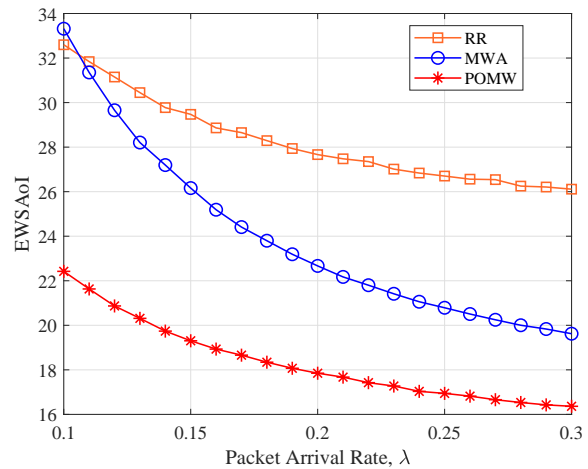}
	\caption{EWSAoI performance versus packet arrival rate $\lambda$ in the PON with $N=10$, and $\omega_i\sim \mathcal{U}(0.1,1.9)$, $p_i\sim \mathcal{U}(0.1,0.9)$, $\forall i$. Note that the notation $\mathcal{U}(a,b)$ denotes a uniform distribution over the real number interval $[a,b]$.}
	\label{Baseline_I}
\end{figure}

Fig. \ref{Baseline_I} plots the long-term EWSAoI curves of all three policies as in Fig. \ref{Baseline_B} over asymmetric PONs.
The transmission successful rates and weight coefficients of all nodes are randomly drawn from uniform distributions in each simulation run.
We can see that in this case, the performance of POMW policy clearly outperforms that of the RR and MWA policies.
Furthermore, the performance of the MWA policy cannot approach to that of the POMW policy even when the arrival rate increases to $0.3$.
This is because
the expected local ages are not symmetric in an asymmetric PON and the MWA policy does not consider this information.

\section{Conclusions}

In this paper, we investigated the AoI-oriented scheduling problem for a wireless multiuser uplink network. Due to the partial observations of the local ages at end devices, we formulated the scheduling decision-making problem as a partially observable Markov decision process (POMDP). The POMDP was first reformulated to an equivalent belief-MDP, and then simplified to an Last-Observation-Characterized (LOC) belief-MDP by adequately leveraging the properties of the status update arrival processes. With the simplification, the infinite dimensional belief states can be characterized by two-dimensional vectors, and thus the complexity of belief updates is significantly reduced. On this basis, we devised the Partially Observable Max-Weight (POMW) policy that minimizes the expected weighted sum AoI of the next slot on condition of the current belief state. Based on the LOC belief-MDP, we derived upper bounds for the performance of the proposed POMW policy. Moreover, we evaluated the performance guarantee of the POMW policy by comparing its performance with a universal lower bound available in the literature. Finally, simulation results validated our analyses, illustrating that the performance gap between the proposed POMW policy and its fully observable counterpart is proportional to the inverse of the lowest arrival rate.
The simulation results also validated the superiority of the POMW policy over the baseline policies.

Future work includes the development of a Whittle's index-based policy for the considered scheduling problem, the extension to the scenarios where the packet arrival rates at end nodes are not known a priori, as well as the extension to more recent information freshness metrics (e.g., AoI at Query \cite{9676636}).

\appendices

\section{Proof of Proposition \ref{prop_km}}\label{appA}
	 \textbf{Proposition \ref{prop_km}} can be proved by induction. First, we show that 
	\begin{equation}\label{Belief_km}
		\bm{c}(k,1)=\left[ \lambda,0,\cdots,0,1-\lambda,0,\cdots\right], 
	\end{equation}
	where $(1-\lambda)$ is the $(k+1)$-th entry if we suppose $\hat{d}_t$ equals to $k$. Then, according to \eqref{eqAoI}, the destination AoI $D_{t+1}=k+1$ because there is no packet received from the node in slot $t$ and $d_t=k$. This satisfies \textbf{Proposition \ref{prop_km}}.
	
	Suppose $\bm{c}(k,m)$ satisfies \eqref{ckm}, i.e.,
	\begin{equation}
		\bm{c}(k,m) =\left[ \lambda,\lambda\gamma,\lambda\gamma^2,\cdots,\lambda\gamma^{m-1},0,\cdots,0,\gamma^m,0,\cdots\right],
	\end{equation}
	where $\gamma^m$ is the $(k+m)$-th entry, and $D_{t+m}=k+m$. 
	Then, on one side, according to \eqref{Trans_d}, we have
	\begin{equation}\label{Belief_km+1}
		\begin{split}
			&\bm{c}(k,m+1)\\
			&=\left[b_{t+m+1}(1),\lambda\gamma,\lambda\gamma^2,\cdots,\lambda\gamma^{m},0,\cdots,0,\gamma^{m+1},0,\cdots\right]
		\end{split}
	\end{equation}
	with $\gamma^{m+1}$ being the $(k+m +1)$-th entry. Because \eqref{Belief_km+1} is a probability distribution, we have
	\begin{equation}
	  \begin{split}
	    b_{t+m+1}(1)&=1-\sum_{d>1}b_{t+m+1}(d)\\
	    &=1-\left(\lambda\sum^m_{l=1}\gamma^l+\gamma^{m+1}\right)\\
	    &=1-\gamma=\lambda.
	  \end{split}
	\end{equation} 
	On the other side, $D_{t+m+1}=k+m+1$, since there is no packet received from the node in slot $t+m$. Thus, the local age belief state and AoI in slot $t+m+1$ still satisfy \textbf{Proposition \ref{prop_km}}. The proposition is proved.

\section{Proof of Theorem \ref{UBTh}}\label{appD}
	According to the definition, the Max-Weight policy minimizes the expected sum AoI of the next slot under the condition of $\bm{B}_t$. Thus, the Randomized Scheduling policy with arbitrary feasible scheduling probabilities $\left\lbrace \mu_i\right\rbrace^N_{i=1} $ yields a higher (or equal) value of the expected sum AoI of the next slot. We then have
	\begin{equation}\label{DriftIeq}
		\begin{split}
			&\frac{1}{N}\sum_{i=1}^{N}\mathbb{E}\left[ \beta_i\left( D_{t+1,i}-D_{t,i}\right)\Bigg|\bm{B}_t \right] \\
			&\le \frac{1}{N}\sum_{i=1}^{N}\beta_i+\frac{1}{N}\sum_{i=1}^{N}\mu_ip_i\beta_i\left( \sum_{d\in\mathcal{D}}b_{t,i}(d)d-D_{t,i}\right),\forall t.
		\end{split} 
	\end{equation} 
	Now, we analyze the RHS of \eqref{DriftIeq}. Given $\lambda_i$, $p_i$, $\beta_i$ and $\bm{B}_t$, according to the LOC belief-MDP, we have
	\begin{equation}\label{kmL}
		\begin{split}
			\sum_{d\in\mathcal{D}}b_{t,i}(d)d&=\frac{1}{\lambda_i}+\left( k_{t,i}-\frac{1}{\lambda _i}\right)\left(1-\lambda_i \right)^{D_{t,i}-k_{t,i}},\\
		\end{split} 
	\end{equation}
	where $k_{t,i}\in \left\lbrace 1,2,\cdots,D_{t,i}-1\right\rbrace $.
	Define 
	\begin{equation}
		X(k)\triangleq\left( k-\frac{1}{\lambda _i}\right)\left(1-\lambda_i \right)^{D_{t,i}-k},
	\end{equation}
	where $k$ denotes a possible value of $k_{t,i}$. Clearly, $\left(1-\lambda_i \right)^{D_{t,i}-k}$ is always positive, we hence discuss the sign of the term $ k-\frac{1}{\lambda _i}$. 
	
	When $k-\frac{1}{\lambda _i}>0$, $X(k-1)<X(k)$ is equivalent to
	\begin{equation}\label{XRatio1}
		\dfrac{X(k-1)}{X(k)}=\left( 1-\lambda_i\right)\left(1-\dfrac{1}{k-1/\lambda_i} \right)<1. 
	\end{equation}
	Manipulating \eqref{XRatio1} and considering $k-\frac{1}{\lambda _i}>0$, we find that \eqref{XRatio1} holds if $k>\frac{1}{\lambda _i}$. As $k\in\mathbb{Z}^+$, $X(k)$ increases on $k=\left\lfloor \frac{1}{\lambda_i}\right\rfloor,\left\lceil \frac{1}{\lambda_i}\right\rceil,\left\lceil \frac{1}{\lambda_i}\right\rceil+1,\cdots$. 
	
	When $k-\frac{1}{\lambda _i}<0$, $X(k-1)<X(k)$ is equivalent to
	\begin{equation}\label{XRatio2}
		\left( 1-\lambda_i\right)\left(1-\dfrac{1}{k-1/\lambda_i} \right)>1. 
	\end{equation}
	Similarly, we obtain that $X(k)$ increases on $k=1,2,\cdots,\left\lfloor \frac{1}{\lambda_i}\right\rfloor$. Therefore, for any given $D_{t,i}>1$, $X(1)<X(2)<\cdots<X(D_{t,i}-1)$. This property also holds when $1/\lambda_i$ is an integer. Substituting $k_{t,i}=D_{1,i}-1$ into \eqref{DriftIeq} and \eqref{kmL} yields
	\begin{equation}\label{DriftIeq2}
		\begin{split}
			&\frac{1}{N}\sum_{i=1}^{N}\mathbb{E}\left[ \beta_i\left( D_{t+1,i}-D_{t,i}\right)\Bigg|\bm{B}_t \right] \\
			&\le \frac{1}{N}\sum_{i=1}^{N}\beta_i+\frac{1}{N}\sum_{i=1}^{N}\mu_ip_i\beta_i\left( \lambda_i-\lambda_iD_{t,i}\right),\forall t.
		\end{split} 
	\end{equation} 
	
	Taking the expectation of both sides of \eqref{DriftIeq2} with respect to $\bm{B}_t$, taking a sum over $t\in\left\lbrace 1,2,\cdots,T\right\rbrace $, and then taking the time-average, we have
	\begin{equation}\label{DriftIeq3}
		\begin{split}
			&\frac{1}{NT}\sum_{i=1}^{N}\mathbb{E}\left[ \beta_iD_{t+1,i}\right]- \frac{1}{NT}\sum_{i=1}^{N}\mathbb{E}\left[ \beta_iD_{1,i}\right]\\
			&\le \frac{1}{N}\sum_{i=1}^{N}\beta_i+\frac{1}{NT}\sum_{i=1}^{N}\sum_{t=1}^{T}\mu_ip_i\beta_i\mathbb{E}\left[ \lambda_i-\lambda_iD_{t,i}\right].
		\end{split} 
	\end{equation}
	Rearranging \eqref{DriftIeq3}, we take the limit as $T\to \infty$ and assign $\beta_i=\omega_i/\lambda_i\mu_ip_i$. Then, we have 
	\begin{equation}\label{UBMW}
		\begin{split}
			\lim_{T\to\infty}\frac{1}{NT}\sum_{i=1}^{N}\sum_{t=1}^{T}\omega_i\mathbb{E}\left[ D_{t,i}\right] &\le \frac{1}{N}\sum_{i=1}^{N}\dfrac{\omega_i}{\lambda_i\mu_ip_i}+\frac{1}{N}\sum_{i=1}^{N}\omega_i\\
			&\overset{(a)}{\le}\frac{1}{N}\sum_{i=1}^{N}\dfrac{\omega_i}{\lambda_i\mu_ip_i}+\frac{1}{N}\sum_{i=1}^{N}\frac{\omega_i}{\lambda_i},
		\end{split}
	\end{equation}
	where the inequality $(a)$ follows as $\lambda_i\in\left(0,1 \right],\forall i$. 
	
	The left hand side (LHS) of \eqref{UBMW} is the long-term EWSAoI of the network under the Max-Weight policy. Following the Lyapunov optimization, to minimize the upper bound, i.e., RHS of \eqref{UBMW}, with respect to $\mu_i$, we equivalently need to solve 
	
	\begin{equation}\label{pmw1}
	    \begin{split}
	    \min\limits_{\left\lbrace \mu_i\right\rbrace^N_{i=1} }\quad & 
		\sum_{i=1}^{N}\dfrac{\omega_i}{\lambda_i\mu_ip_i} ,\\
		\mbox{s.t.}\quad & \sum_{i=1}^{N}\mu_i\le 1. 
	    \end{split}
	\end{equation}
	
	By applying similar procedures presented in the proof of \cite[Th. 5]{8933047}, we can obtain the optimal solution to the problem \eqref{pmw1} given by
	\begin{equation}\label{mu'}
		\mu'_i=\dfrac{\sqrt{\omega_i/\lambda_ip_i}}{\sum_{j=1}^{N}\sqrt{\omega_j/\lambda_jp_j}}, \forall i.
	\end{equation}
	Substituting \eqref{mu'} makes the RHS of \eqref{UBMW}, we can arrive at \eqref{mu'} after some necessary manipulations. 
	This completes the proof.
	

\section{Proof of Corollary \ref{RLR} }\label{appE}

	According to the constraint \eqref{LB2}, $\mu_i=q^*_i/p_i$ satisfies the inequality \eqref{DriftIeq}. Thus, substituting $\mu_i=q^*_i/p_i$ and $\beta_i=\omega_i/\lambda_i\mu_ip_i$ into \eqref{UBMW} gives 
	\begin{equation}
		\begin{split}
			R^{POMW}&\le\frac{1}{N}\sum_{i=1}^{N}\omega_i\left(\frac{1}{\lambda_iq^*_i}+1 \right)\\
			&< \frac{1}{N}\sum_{i=1}^{N}\omega_i\left(\frac{1}{\lambda_{min}q^*_i}+\frac{3}{\lambda_{min}} \right)=\dfrac{2L_B}{\lambda_{min}},
		\end{split}
	\end{equation}
where the second inequality follows as $\lambda_{min}\le\lambda_i\le1$, and this completes the proof.

\section{Extension to Markovian Arrivals}\label{appF}

We extend our framework to the scenario with Markovian arrival processes in this appendix.
The packet arrival process at each node is characterized by a Markov chain, where the probability of a packet arrival in the current slot depends on the packet arrival situation in the previous slot. 
Let $A_{t,i}\in\left\{0,1\right\}$ denote the indicator of the packet arrival situation of node $i$ in slot $t$,
$A_{t,i}=1$ when there is a packet arrival at node $i$ in slot $t$, and $A_{t,i}=0$ otherwise.
The transition functions of $A_{t,i}$ can be written as
\begin{equation}\label{Trans_d1}
	\Pr\left( A_{t+1,i}|A_{t,i}\right)=
	\begin{cases}
		\overline{\Lambda}_i, &\text{if}\ A_{t,i}=1,A_{t+1,i}=1,\\
		\overline{\Gamma}_i, &\text{if}\ A_{t,i}=1,A_{t+1,i}=0,\\
		\Lambda_i, &\text{if}\ A_{t,i}=0,A_{t+1,i}=1,\\
		\Gamma_i, &\text{if}\ A_{t,i}=0,A_{t+1,i}=0,\\
		0, &\text{otherwise},
	\end{cases} 
\end{equation}
where $\Gamma_i \triangleq 1-\Lambda_i$ and $\overline{\Gamma}_i \triangleq 1-\overline{\Lambda}_i$.
The Markov chain diagram of the status update arrivals is shown in Fig. \ref{MPA}.
\begin{figure}[H]
	\centering
	\includegraphics[width=0.4\textwidth]{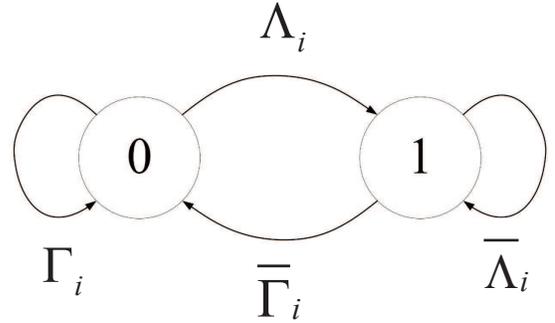}
	\caption{The Markov chain of the status update arrival process at node $i$.}
	\label{MPA}
\end{figure}
Recall that the status update arrival situation of each time slot is known to each node at the end of the said slot.
In this case, the update packet transmitted from the scheduled node to the AP only contains the status update arrived in the previous slot.
We note that the value of the local age of each node in the current slot implies the packet arrival situation in the previous time slot. 
Specifically, $d_{t+1,i}=1$ represents that a status update arrived at node $i$ in slot $t$, and $d_{t+1,i}>1$ otherwise. 
Therefore, the transition functions of the local age can be expressed as
\begin{equation}\label{Trans_dM}
	\Pr\left( d_{t+1,i}|d_{t,i}\right)=
	\begin{cases}
		\overline{\Lambda}_i, &\text{if}\ d_{t,i}=1 \text{ and } d_{t+1,i}=1, \\
		\overline{\Gamma}_i, &\text{if}\ d_{t,i}=1 \text{ and }d_{t+1,i}=2, \\
		\Lambda_i, &\text{if}\ d_{t,i}>1 \text{ and }d_{t+1,i}=1, \\
		\Gamma_i, &\text{if}\ d_{t,i}>1 \text{ and }d_{t+1,i}=d_{t,i}+1, \\
		0, & \text{otherwise.}
	\end{cases} 
\end{equation}

To characterize the belief state of the local age of a node mathematically, we artificially introduce a local observer at each node that can decide to observe whether a status update arrives at the said node or not.
Note that such observers do not exist in practice, and are introduced to facilitate the characterization of the belief state of the local age of a node.
The observation of the arrival of the status update at node $i$ in slot $t$ is denoted by $\hat{A}_{t,i}\in\left\{0,1,X\right\}$, where $\hat{A}_{t,i} = 1$ when an arrival of a status update is observed at node $i$ in slot $t$, $\hat{A}_{t,i} = 0$ when no arrival is observed at node $i$ in slot $t$, and $\hat{A}_{t,i} = X$ when no observation is made at node $i$ in slot $t$.
Let $\omega_{t,i}\triangleq \Pr(A_{t,i}=1|\bm{h}_{t,i})$ denotes the belief probability that a new status update arrives at node $i$ in slot $t$ given $\bm{h}_{t,i}\triangleq \left\langle \omega_{1,i},\hat{A}_{1,i},\dots,\hat{A}_{t-1,i}\right\rangle$. Hence, the belief state of the arrival of the status update at node $i$ can be expressed as $[\Pr(A_{t,i}=0|\bm{h}_{t,i}),\Pr(A_{t,i}=1|\bm{h}_{t,i})]=[1-\omega_{t,i},\omega_{t,i}]$, which is a two-dimension simplex. 
Given $\hat{A}_{t,i}$, $\omega_{t+1,i}$ can be updated by
\begin{equation}
    \omega_{t+1,i}=\eta(\omega_{t,i},\hat{A}_{t,i})=
    \begin{cases}
        \overline{\Lambda}_i,&\text{if}\ \hat{A}_{t,i}=1,\\
        \Lambda_i,&\text{if}\ \hat{A}_{t,i}=0,\\
        \mathcal{T}(\omega_{t,i}),&\text{if}\ \hat{A}_{t,i}=X,
    \end{cases}
\end{equation}
where $\mathcal{T}(\omega_{t,i})=\omega_{t,i}\overline{\Lambda}_i+(1-\omega_{t,i})\Lambda_i$ denotes the one step belief update of $\omega_{t,i}$. Additionally, we define the one step belief update of $1-\omega_{t,i}$ as $\mathcal{G}(1-\omega_{t,i})\triangleq 1-\mathcal{T}(\omega_{t,i})$.
Let $\mathcal{T}^m(\omega_{t,i})\triangleq \Pr(A_{t+m,i}=1|\omega_{t,i})$ denote the $m$-step belief update of $\omega_{t,i}$ when the arrival situation is unobserved for $m$ consecutive slots, where $m\in\left\{0,1,\cdots\right\}$ and $\mathcal{T}^0(\omega_{t,i})=\omega_{t,i}$. $\mathcal{G}^m(1-\omega_{t,i})\triangleq \Pr(A_{t+m,i}=0|\omega_{t,i})=1-\mathcal{T}^m(\omega_{t,i})$ denotes the $m$-step belief update of $1-\omega_{t,i}$ in the same case, where $\mathcal{G}^0(1-\omega_{t,i})=1-\omega_{t,i}$.
Further,
\begin{equation}\label{omgupdate}
    \begin{split}
  \mathcal{T}^m(\omega)&=\frac{\Lambda-(\overline{\Lambda}-\Lambda)^m(\Lambda-(1+\Lambda-\overline{\Lambda})\omega)}{1+\Lambda+\overline{\Lambda}}\\
  &\in [0,1],{\ \ \ \ }\forall \omega\in [0,1].
    \end{split}
\end{equation}
We remark that we follow a method in \cite{5605371} to derive \eqref{omgupdate} by formulating a partially observable two-state Markov chain.
For brevity, we refer the readers to the proof of \textbf{\textit{Lemma} 1} of \cite{5605371} for the derivation details.

We subsequently show the existence of a simplified representation of the belief states of the local age with the given $\bm{B}_1$. 
To start, we
have the following definition:
\begin{definition}\label{DefekmMa}
    Assume AP schedules node $i$ in slot $t$ with observation $d_{t,i}=k_i$, and then does not receive
	any packet from node $i$ in the following $m_i$ slots. Define the local age belief state of node $i$ in slot $t+m_i$ by $\bm{e}(k_i,m_i)$, namely, the local age belief of node $i$ with the last observation $k_i$ followed by $m_i$ elapsed slots.    
\end{definition}

For convenience, we ignore index $i$ for nodes and introduce the following proposition to show the simplified representation of the belief state of the type given in \textbf{\textit{Definition} \ref{DefekmMa}}.
\begin{proposition}\label{prop_kmM}
	The distribution vector of the local age belief state $\bm{e}(k,m)$ of node $i$ in slot $t$ can be given by \eqref{e1m} and \eqref{ekm2}  given on top of next page,
	\begin{figure*}
	    \begin{equation}\label{e1m}
		\begin{split}
			&\bm{e}(1,m)=\left[ e_{1,m}\left( d_t\right) \right]_{d_t\in\mathbb{Z}^+}\\
			&=\left[ \mathcal{T}^{m-1}(\overline{\Lambda}),\mathcal{G}^{m-1}(\overline{\Gamma})\mathcal{T}^{m-2}(\overline{\Lambda}),\mathcal{G}^{m-1}(\overline{\Gamma})\mathcal{G}^{m-2}(\overline{\Gamma})\mathcal{T}^{m-3}(\overline{\Lambda}),
		\cdots,\overline{\Lambda}\prod^{m-1}_{d=1}\mathcal{G}^d(\overline{\Gamma}),\prod^{m-1}_{d=0}\mathcal{G}^d(\overline{\Gamma}),0,\cdots\right],
		\end{split} 
	\end{equation}
	\begin{equation}\label{ekm2}
		\begin{split}
			&\bm{e}(k>1,m)=\left[ e_{k,m}\left( d_t\right) \right]_{d_t\in\mathbb{Z}^+}\\
			&=\left[ \mathcal{T}^{m-1}(\Lambda),\mathcal{G}^{m-1}(\Gamma)\mathcal{T}^{m-2}(\Lambda),\mathcal{G}^{m-1}(\Gamma)\mathcal{G}^{m-2}(\Gamma)\mathcal{T}^{m-3}(\Lambda),
		\cdots
		,\Lambda\prod^{m-1}_{d=1}\mathcal{G}^d(\Gamma),0,\cdots,0,\prod^{m-1}_{d=0}\mathcal{G}^d(\Gamma),0,\cdots\right].
		\end{split} 
	\end{equation}
	\end{figure*}
	where $k,m \in \mathbb{Z}^+ $, and $e_{k,m}\left( {d}_t\right) $ denotes the belief probability assigned to ${d}_t$. The position of the entry $\prod^{m-1}_{d=0}\mathcal{G}^d(\Gamma)$ is $k+m$,
	indicating that the corresponding destination AoI of entry $\prod^{m-1}_{d=0}\mathcal{G}^d(\Gamma)$ is $k+m$.
\end{proposition}

\begin{proof}
    \textbf{\textit{Proposition} \ref{prop_kmM}} can be proved by induction. First, we show that 
    \begin{equation}
        \bm{e}(1,1)=\left[ \overline{\Lambda},\overline{\Gamma},0,\cdots,0,\cdots\right].
    \end{equation}
    Assume the corresponding slot of this belief state is $t+1$.
    $k=1$ and $m=1$ imply that the AP observed $d_{t}=1$. 
    According to the Markovian arrival process, we have $A_{t-1}=1$, which represents that the probability of $A_{t}=1$ is $\overline{\Lambda}$.
    Thus, by \eqref{Trans_dM}, the probability of $d_{t+1}=1$ is $\overline{\Lambda}$ and the probability of $d_{t+1}=d_{t}+1=2$ (i.e., no packet arrival in slot $t$) is $\overline{\Gamma}$.
    This satisfies \textbf{\textit{Proposition} \ref{prop_kmM}}.
    
    Suppose $\bm{e}(1,m)$ satisfies \eqref{e1m}.
    We first verify that $\bm{e}(1,m)$ is a probability distribution.
    Following the fact that  $\mathcal{T}^d(\overline{\Lambda})+\mathcal{G}^d(\overline{\Gamma})=1,\forall d\in\mathbb{N}$, the summation of the non-zero entries of $\bm{e}(1,m)$ from the second one to the last one can be derived and simplified as follows
    \begin{equation}
        \begin{split}
            &\mathcal{G}^{m-1}(\overline{\Gamma})\mathcal{T}^{m-2}(\overline{\Lambda})+\cdots+\mathcal{T}(\overline{\Lambda})\prod^{m-1}_{d=2}\mathcal{G}^d(\overline{\Gamma})\\
            &+\overline{\Lambda}\prod^{m-1}_{d=1}\mathcal{G}^d(\overline{\Gamma})+\prod^{m-1}_{d=0}\mathcal{G}^d(\overline{\Gamma})\\
            &= \mathcal{G}^{m-1}(\overline{\Gamma})\mathcal{T}^{m-2}(\overline{\Lambda})+\cdots+\mathcal{T}(\overline{\Lambda})\prod^{m-1}_{d=2}\mathcal{G}^d(\overline{\Gamma})+\prod^{m-1}_{d=1}\mathcal{G}^d(\overline{\Gamma})\\
            &=\mathcal{G}^{m-1}(\overline{\Gamma})\mathcal{T}^{m-2}(\overline{\Lambda})+\cdots+\prod^{m-1}_{d=2}\mathcal{G}^d(\overline{\Gamma})\\
            &\cdots\\
            &=\mathcal{G}^{m-1}(\overline{\Gamma})\mathcal{T}^{m-2}(\overline{\Lambda})+\mathcal{G}^{m-1}(\overline{\Gamma})\mathcal{G}^{m-2}(\overline{\Gamma})\\
            &=\mathcal{G}^{m-1}(\overline{\Gamma}).
        \end{split}
    \end{equation}
    It becomes clear that $\left\|\bm{e}(1,m)\right\|_1=1$ and $e_{1,m}(d)\in [0,1],\forall d\in \mathbb{Z}^+$ , i.e., $\bm{e}(1,m)$ is a distribution.
    According to \textbf{\textit{Definition} \ref{DefekmMa}}, $\bm{e}(1,m)$ reveals the expression of the belief state of the local age of a node when the AP observed a status update arrives $m$ consecutive unobserved slots after slot $t$ (i.e., in slot $t+m$).
    Hence, in slot $t+m$, the belief probabilities of a packet arrival and no packet arrival are $\mathcal{T}^m(\overline{\Lambda})$ and $\mathcal{G}^m(\overline{\Gamma})$, respectively.
    
    Now, we consider $\bm{e}(1,m+1)$, i.e., the local age belief state in slot $t+m+1$.
    The value of $d_{t+m}$ could be equal to one of the positions of the non-zero entries in $\bm{e}(1,m)$.
    Despite of the value of $d_{t+m}$, $d_{t+m}$ transits to $d_{t+m+1}=1$ with probability $\mathcal{T}^m(\overline{\Lambda})$.
    As such, $e_{1,m+1}\left( 1\right)=\left\|\bm{e}(1,m)\right\|_1\mathcal{T}^m(\overline{\Lambda})=\mathcal{T}^m(\overline{\Lambda})$.
    On the other hand, the probability $d_{t+m}$ transits to $d_{t+m+1}=d_{t+m}+1$ is $\mathcal{G}^m(\overline{\Gamma})$, and hence $e_{1,m+1}\left( d\right)=e_{1,m}\left( d-1\right)\mathcal{G}^m(\overline{\Gamma}),\forall d>1$.
    Thus, the expression of $\bm{e}(1,m+1)$ is given by \eqref{e1m+1},
    \begin{figure*}
    \begin{equation}\label{e1m+1}
        \begin{split}
			\bm{e}(1,m+1)
			=\left[ \mathcal{T}^{m}(\overline{\Lambda}),\mathcal{G}^{m}(\overline{\Gamma})\mathcal{T}^{m-1}(\overline{\Lambda}),\mathcal{G}^{m}(\overline{\Gamma})\mathcal{G}^{m-1}(\overline{\Gamma})\mathcal{T}^{m-2}(\overline{\Lambda}),
		\cdots,\overline{\Lambda}\prod^{m}_{d=1}\mathcal{G}^d(\overline{\Gamma}),\prod^{m}_{d=0}\mathcal{G}^d(\overline{\Gamma}),0,\cdots\right].
		\end{split}
    \end{equation}
    \hrulefill
    \end{figure*}
    which follows \textbf{\textit{Proposition} \ref{prop_kmM}}.
    
    Similarly, we can also prove \eqref{ekm2}, which is omitted for brevity.
    This completes the proof.

\end{proof}

The belief state in \textbf{\textit{Definition} \ref{DefekmMa}} is the distribution of the local age only, while the completed belief state of a node also involves the destination AoI.
To this end, we define a group of belief states that have the destination AoI equal to $k+m$ together with $\bm{e}(k,m)$ defined in \textbf{\textit{Proposition} \ref{prop_kmM}} as $\bm{E}(k,m)\triangleq \left\langle k+m,\bm{e}(k,m)\right\rangle$ for $m,k\in \mathbb{Z}^+$. 
Denote by $\bm{\mathcal{E}}$ the collection of all possible $\bm{E  }(k,m)$. Then, we have the following corollary.
\vspace{4mm}
\begin{corollary}\label{cokmM}
 Suppose the network with the MAP has a belief state $\bm{B}_{t,i}\in \bm{\mathcal{E}}$, then $\bm{B}_{t',i}\in \bm{\mathcal{E}}$ for any $t'>t$.   
\end{corollary}

\vspace{4mm}

The proof of \textbf{\textit{Corollary} \ref{cokmM}} is similar to that of \textbf{\textit{Corollary} \ref{corokm}}, and hence is omitted for brevity.

\ifCLASSOPTIONcompsoc
 \section*{Acknowledgments}
\else
 \section*{Acknowledgment}
\fi

The authors would like to thank Tong Zhang and Yijin Zhang for their helpful discussions on establishing the network model and problem formulation. 
%
%
%
%
%
%

\bibliographystyle{IEEEtran}
\bibliography{Ref}

\end{document}